\RequirePackage{fix-cm}
\documentclass[twocolumn]{svjour3}          
\smartqed  
\usepackage{listings}
\usepackage{xspace}
\usepackage{subcaption}
\usepackage{rotating}
\usepackage{caption}
\usepackage{array}
\usepackage{booktabs}
\usepackage{multirow}
\usepackage{caption} 
\usepackage{enumitem}
\usepackage{amsmath}    
\usepackage{amssymb}    
\usepackage{xcolor}     
\usepackage{hyperref}   

\newcommand{\kw}[1]{{\ensuremath {\mathsf{#1}}}\xspace}

\newcommand{\revise}[1]{\textcolor{black}{#1}}

\newcommand{\revisesigmod}[1]{\textcolor{black}{#1}}

\newtheorem{definition}{Definition}
\newtheorem{example}{Example}

\newcommand{\sys}{\kw{ExBI}}
\newcommand{\sysnosample}{\kw{ExBI_{cmpl}}}
\newcommand{\explorativebi}{\kw{Exploratory~BI}}
\newcommand{\source}{\kw{Source}}
\newcommand{\join}{\kw{\bowtie}}

\newcommand{\view}{\kw{View}}
\newcommand{\sourcenosample}{\kw{SOURCE_{cmpl}}}
\newcommand{\sourcesample}{\kw{SOURCE_{fas}}}
\newcommand{\joinnosample}{\kw{JOIN_{cmpl}}}
\newcommand{\joinsample}{\kw{JOIN_{sgs}}}

\newcommand{\sourceop}{\kw{Source}}
\newcommand{\joinop}{\kw{Join}}
\newcommand{\viewop}{\kw{View}}
\newcommand{\drilldownop}{\kw{DrillDown}}
\newcommand{\sliceop}{\kw{Slice}}
\newcommand{\rollupop}{\kw{RollUp}}
\newcommand{\diceop}{\kw{Dice}}

\newcommand{\sumfunc}{\texttt{SUM}\xspace}
\newcommand{\countfunc}{\texttt{COUNT}\xspace}
\newcommand{\distinctcountfunc}{\texttt{DISTINCT~COUNT}\xspace}
\newcommand{\maxfunc}{\texttt{MAX}\xspace}
\newcommand{\minfunc}{\texttt{MIN}\xspace}

\newcommand{\quantilefunc}{\texttt{QUANTILE}\xspace}
\newcommand{\prop}{\kw{prop}}

\long\def\comment#1{}

\newcommand{\stitle}[1]{\noindent{\underline{ #1}}}

\newcommand{\reffig}[1]{Fig.~\ref{fig:#1}}
\newcommand{\refsec}[1]{Sec.~\ref{sec:#1}}
\newcommand{\reftable}[1]{Table~\ref{tab:#1}}

\newcommand{\refdef}[1]{Def.~\ref{def:#1}}

\newcommand{\reflem}[1]{Lemma~\ref{lemma:#1}}

\definecolor{eclipseBlue}{RGB}{42,0.0,255}
\definecolor{eclipseGreen}{RGB}{63,127,95}
\definecolor{eclipsePurple}{RGB}{127,0,85}

\lstdefinelanguage{ebi}
{
  morekeywords={
    SOURCE,
    DRILLDOWN,
    ROLLUP,
    SLICE,
    DICE,
    JOIN,
    ON,
    AS,
    WHERE,
    ORDER,
    BY,
    LIMIT,
    OFFSET,
    GROUP,
    NODE,
    EDGE,
    TYPE,
    IMPORTS,
    OPTIONAL,
    STRICT,
    LOOSE,
    OPEN,
    CLOSED,
    ABSTRACT,
    CREATE,
    GRAPH,
    FOR,
    WITHIN,
    MATCH,
    WHERE,
    NOT,
    OR,
    AND,
    EXISTS,
    RETURN,
    IN,
    IS,
    NULL,
    ORDER,
    BY,
    INSERT,
    SET,
    REMOVE,
    DROP,
    DELETE,
    FINISH,
    COLUMNS,
    GRAPH_TABLE,
    SELECT,
    DETACH,
    STRING,
    VARCHAR,
    BOOLEAN,
    BOOL,
    SIGNED,
    INTEGER,
    INT,
    FLOAT,
    GRAPH,
    TYPE,
    ALTER,
    ROLLBACK,
    COMMIT,
    TRANSACTION,
    START,
    SESSION,
    USE,
    GROUP,
    VALUE,
    COUNT,
    CALL,
    ANY,
    IMPLIES,
    SCHEMA,
    AT,
    AS,
    FROM,
    GRAPH\_TABLE,
    INTERSECT,
    \$,
    KEY,
    TABLES,
    TABLE,
    SOURCE,
    DESTINATION,
    REFERENCE,
    PROPERTY,
    PROPERTIES,
    VERTEX,
    JOIN,
    ON
  },
  sensitive=false, 
  morecomment=[l]{//}, 
  morecomment=[s]{/*}{*/}, 
  morestring=[b]" 
}

\begin{document}

\title{A Hypergraph-Based Framework for Exploratory Business Intelligence\thanks{Longbin Lai is the corresponding author.\\ Shunyang Li and Jianke Yu contributed to this work during their internship at Alibaba.}}



\author{Yunkai Lou         \and
        Shunyang Li \and
        Longbin Lai \and
        Jianke Yu \and
        Wenyuan Yu \and
        Ying Zhang
}


\institute{Yunkai Lou \at
          Alibaba Group, Hangzhou, China \\
          \email{louyunkai.lyk@alibaba-inc.com}           
           \and
           Shunyang Li \at
              University of New South Wales, Sydney, Australia \\
              \email{shunyang.li@unsw.edu.au}
           \and
           Longbin Lai \at
              Alibaba Group, Hangzhou, China \\
              \email{longbin.lailb@alibaba-inc.com}
           \and
           Jianke Yu \at 
           University of Technology Sydney, Sydney, Australia \\
           \email{jianke.yu@student.uts.edu.au}
           \and
           Wenyuan Yu \at
              Alibaba Group, Hangzhou, China \\
              \email{wenyuan.ywy@alibaba-inc.com}
           \and
           Ying Zhang \at
              Zhejiang Gongshang University, Hangzhou, China \\
              \email{ying.zhang@zjgsu.edu.cn}
}

\date{Received: date / Accepted: date}

\maketitle

\begin{abstract}
Business Intelligence (BI) analysis is evolving towards \explorativebi, an iterative, multi-round exploration paradigm where analysts progressively refine their understanding. However, traditional BI systems impose critical limits for \explorativebi: heavy reliance on expert knowledge, high computational costs, static schemas, and lack of reusability. 
We present \sys, a novel system that introduces the \textbf{hypergraph data model} with operators, including \sourceop, \joinop, and \viewop, to enable dynamic schema evolution and materialized view reuse. 
\revisesigmod{Using} \textbf{sampling-based algorithms with provable estimation guarantees}, \sys addresses the computational bottlenecks, while maintaining analytical accuracy. 
Experiments on LDBC datasets demonstrate that \sys achieves significant speedups over existing systems: on average 16.21$\times$ (up to 146.25$\times$) compared to Neo4j and 46.67$\times$ (up to 230.53$\times$) compared to MySQL, while maintaining high accuracy with an average error rate of only 0.27\% for \countfunc, enabling efficient and accurate large-scale exploratory BI workflows.

\keywords{Exploratory Business Intelligence, Hypergraph, Sampling}

\end{abstract}

\section{Introduction}
\label{sec:intro}

Business Intelligence (BI) remains the mainstream data analysis technology \cite{grabova2010business,duan2012business}, with widespread applications and numerous BI tools \cite{powerbi,tableau,oraclebi}. However, as data landscapes become increasingly complex and interconnected, BI analysis scenarios are evolving beyond traditional approaches. Modern analysts often need multi-round iterative exploration, where each step builds upon previous insights and progressively refines understanding of the problem. 

This exploratory analysis paradigm has gained increasing prominence across various domains. For instance, social scientists have employed such iterative approaches to investigate the differential impacts of generative AI on employment \cite{lichtinger2025generative}. In this paper, we refer to such analytical scenarios as \explorativebi.
To illustrate the characteristics and challenges of \explorativebi, we present a concrete example based on the OpenAIRE dataset \cite{openaire}, a comprehensive knowledge graph containing scientific publications, organizations, funding agencies, and their relationships.

\begin{example}
\label{ex:traditional-bi}
We want to analyze the impact of the Russia-Ukraine conflict on Russian scientific research. 

\textbf{Step 1:} We first examine whether Russian institutions' publication counts have declined. Analysts must communicate with database administrators to join the Publication and Organization tables (\reffig{intro:traditional_bi_workflow_step1}). The analysis reveals a significant decline starting in 2022.

\textbf{Step 2:} To investigate causes, we hypothesize funding reductions. This requires initiating a new analysis workflow with administrators to join the Funding table (\reffig{intro:traditional_bi_workflow_step2}). Results show that funding from the European Commission ceased after 2022.

\textbf{Step 3:} We further explore whether funding agencies stopped supporting Russian institutions. This requires yet another cycle with administrators to perform a self-join on Organization (\reffig{intro:traditional_bi_workflow_step3}). The analysis confirms that the European Commission no longer support Russian institutions, which is confirmed by the European Commission's official website~\cite{european_commission_website}.
\end{example}

\begin{figure}[b]
    \centering
    \begin{subfigure}[b]{0.45\textwidth}
        \includegraphics[width=\textwidth]{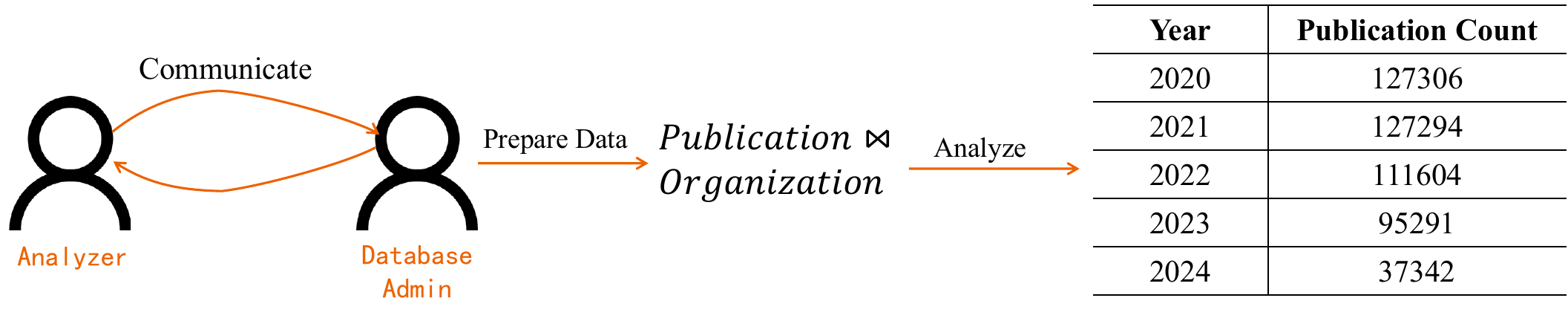}
        \caption{Step 1: Whether the number of Publication by Russian institutions is affected}
        \label{fig:intro:traditional_bi_workflow_step1}
    \end{subfigure}
    \hfill
    \begin{subfigure}[b]{0.45\textwidth}
        \includegraphics[width=\textwidth]{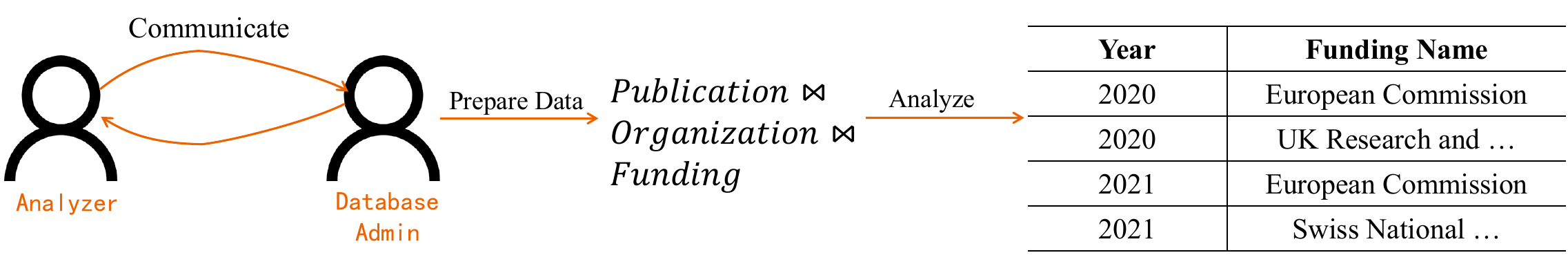}
        \caption{Step 2: Whether the reduction is due to decreased funding}
        \label{fig:intro:traditional_bi_workflow_step2}
    \end{subfigure}
    \hfill
    \begin{subfigure}[b]{0.45\textwidth}
        \centering
        \includegraphics[width=.9\textwidth]{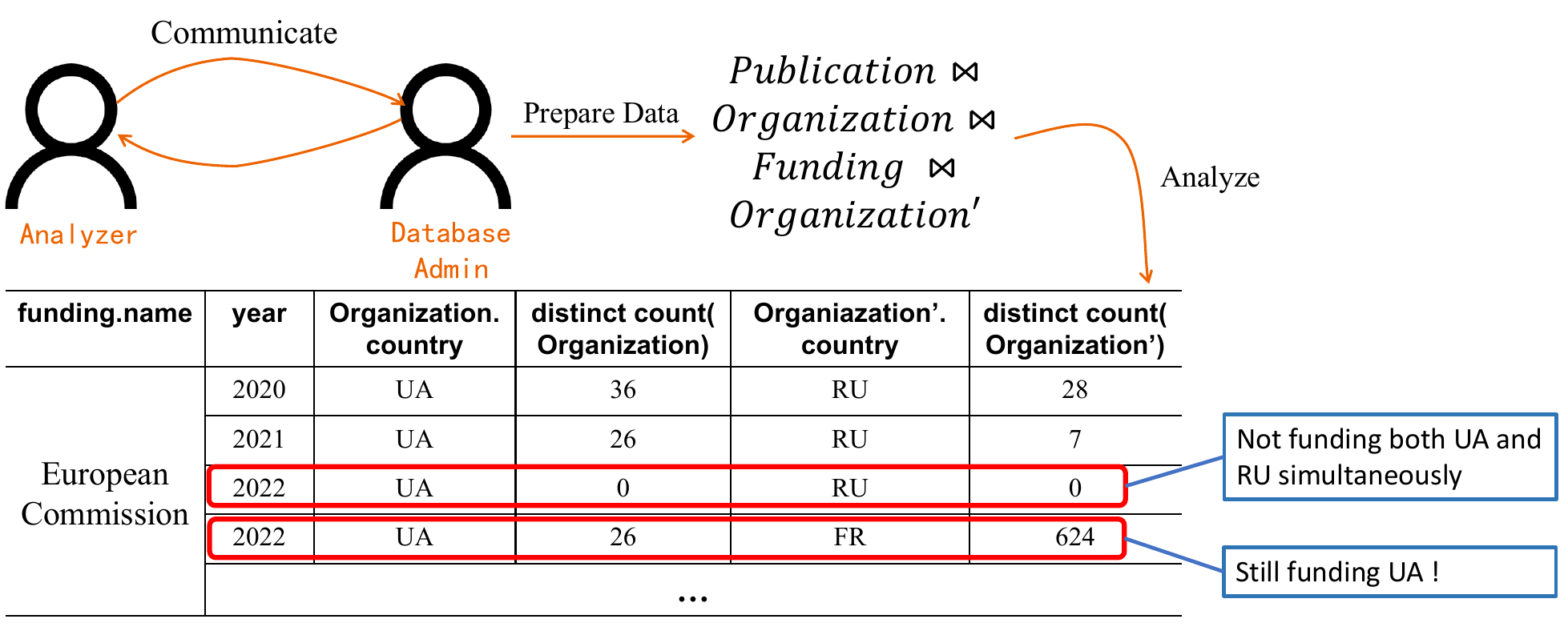}
        \caption{Step 3: Whether the funding agencies support Ukraine but no longer support Russia amid the conflict}
        \label{fig:intro:traditional_bi_workflow_step3}
    \end{subfigure}
    \caption{Analyze the impact of the Russia-Ukraine conflict on Russian scientific research using traditional BI}
    \label{fig:intro:traditional_bi_workflow}
\end{figure}

\subsection{Motivations}
The above example demonstrates four critical limitations of existing BI tools, \revisesigmod{grouped} into \textbf{Logical Representation Gaps} (L1-L3) and \textbf{Physical System Bottlenecks} (L4):

\textbf{(L1) Heavy Reliance on Analyst Expertise.} 
Analysts must identify all relevant data upfront and construct fact/dimension tables, heavily relying on domain knowledge~\cite{chaudhuri2011overview,larson2016review}. \revisesigmod{This is particularly challenging when data involves \textbf{complex n-ary relationships} (e.g., the web of authors, papers, and funds)}, where missing a critical link (e.g., Funding in Step 2) requires restarting the workflow.

\textbf{(L2) Rigid Schema.} \revisesigmod{Star/snowflake schemas~\cite{kimball2013data,chaudhuri2011overview} rely on fixed relational structures that are ill-suited for evolving analysis. Discovering new factors requires rebuilding the data model, and expressing seemingly simple semantic constraints (e.g., a funding that supports Ukraine but \textit{not} Russia) demands complex, unnatural join sequences (e.g., ``non-existence'' relationship check)}.

\textbf{(L3) Lack of Reusability.} Current models lack mechanisms to represent intermediate exploration states as first-class citizens, preventing reuse across iterations and comparison with later findings. In Example~\ref{ex:traditional-bi}, the Publication-Organization join was repeated three times.

\textbf{(L4) Computational Cost.} Building wide tables through table joins is time-consuming and resource-intensive~\cite{orlovskyi2020business,mishra1992join}, making exploratory analysis impractical on large datasets. In Example~\ref{ex:traditional-bi}, Step 3 joins four tables, which becomes prohibitively expensive on the OpenAIRE dataset (193M publications, 51M edges).

Although existing approaches address specific aspects of these challenges, they lack a holistic solution for \explorativebi. Recent LLM-based approaches~\cite{lian2024chatbi,jiang2025siriusbi,weng2025datalab} have attempted to lower the expertise barrier (L1) by directly translating natural language analytical requirements into query statements. However, their performances still face significant gaps before practical deployment. More importantly, these AI-driven solutions merely serve as interfaces and do not address the fundamental system-level challenges (L2-L4) in the underlying BI systems.

Similarly, mainstream BI platforms like Power BI~\cite{powerbi} and Tableau~\cite{tableau} attempt to mitigate schema rigidity (L2) by enabling ad-hoc view creation via joins. However, this flexibility comes with severe scalability and usability constraints. Executing complex multi-way joins on the fly often incurs prohibitive latency (L4), rendering interactive exploration impractical. While pre-materializing these joins as views can improve query response, it results in a static collection of artifacts that cannot adapt to the fluid, unanticipated nature of exploratory analysis. Furthermore, defining the correct set of materialized views requires the very domain expertise (L1) that \explorativebi seeks to democratize.

Therefore, addressing the challenges of \explorativebi requires a fundamentally new approach at the system level, including both the logical and physical aspects, rather than merely adding intelligent interfaces or applying workarounds to existing architectures.



\subsection{Our Approach}
In this paper, we propose \sys, a novel BI system specifically designed to support the \explorativebi scenarios. \revisesigmod{\sys decouples the logical exploration model from the physical storage layer to address the aforementioned challenges comprehensively.}

\textbf{Logically}, \sys addresses the expressiveness and flexibility challenges (\textbf{L1-L3}) by introducing a \textit{Hypergraph Data Model}. 
We adopt this graph-centric approach because \explorativebi scenarios involve complex entity relationships that are ill-suited for rigid relational schemas.
Built upon a logical property graph abstraction, this model allows analysts to define \textit{query graphs} as templates. \sys employs a \sourceop operator to perform \textbf{induced subgraph matching} based on these templates, capturing even complex semantics like ``non-existence'' relationships naturally without cumbersome join logic.
The resulting hypergraphs support dynamic schema evolution (\textbf{L2}) via the \joinop operator and enable the materialization of intermediate results as reusable views (\textbf{L3}). This compositional approach empowers analysts to explore incrementally, reducing the need for upfront expert schema design (\textbf{L1}).

\textbf{Physically}, \sys tackles the computational bottlenecks (\textbf{L4}) through a specialized, sampling-based execution engine. We develop this dedicated engine driven by two key observations. 
First, the core \sourceop operator requires subgraph matching, where performing \textit{exact} computation is infeasible on large-scale datasets --- whether executed on relational engines (via join transformation~\cite{lou2024towards}) or native graph engines (as validated in \refsec{experiment}). Second, although sampling-based tools exist for relational databases (e.g., VerdictDB~\cite{verdictdb}), they handle only relatively simple join patterns (limited depth or without ``non-existence'' semantics~\cite{verdictdb}), failing to meet the complex subgraph matching needs of \explorativebi. 

Consequently, \sys currently implements a custom storage and computation layer to enable efficient approximate processing.
To implement this layer, we design specific sampling strategies for our algebraic operators. For the costly \sourceop operator, we adapt the FaSTest sampling algorithm \cite{FaSTest} to uniformly sample matches instead of exhaustive enumeration. For the \joinop operator, we employ stratified group sampling to handle massive intermediate results.
Crucially, these techniques provide identifiable statistical guarantees, ensuring that \sys achieves high performance without compromising analytical accuracy.
Note that our architecture decouples the execution logic from storage (\refsec{overview}); future work will focus on porting our sampling layer to existing relational and graph databases to reduce migration overhead for users.

We provide a detailed demonstration of how \sys works through the Russia-Ukraine conflict analysis in Section~\ref{sec:exp:case-study}, illustrating the complete workflow including how the hypergraph schema evolves dynamically through operator composition and how intermediate results are preserved and reused across analysis rounds.
Note that our design aligns with the recent agent-first database system proposal~\cite{liu2025agentfirst}. Our incremental exploration framework directly addresses key challenges of agentic speculation, particularly redundancy,  heterogeneity and steerability. Building agentic \explorativebi systems atop \sys is an important future work.


\subsection{Contributions}
This paper makes the following contributions:

(1) \textbf{Exploratory BI Paradigm.} We analyze the limitations of traditional BI in supporting exploratory scenarios and formalize the concept of \explorativebi. By shifting the paradigm from one-shot comprehensive analysis to iterative, incremental exploration, we reduce dependency on analyst expertise and enable data-driven discovery through progressive refinement.

(2) \textbf{Hypergraph Data Model and Operators.} To support the dynamic nature of \explorativebi, we introduce the Hypergraph data model along with \sourceop, \joinop, \viewop, \drilldownop, \rollupop, \sliceop, and \diceop operators (\refsec{model}). This model enables dynamic schema evolution, reusable intermediate views, and operator composition for incremental exploration, providing a theoretical foundation for exploratory BI analysis.

(3) \textbf{System Implementation with Theoretical Guarantees.} We develop \sys, a theoretically-grounded \explorativebi system (\refsec{overview}) that addresses the performance bottlenecks of \sourceop and \joinop operators through sampling-based optimizations. 
\revisesigmod{Critically, \textbf{end-to-end theoretical guarantees for unbiased estimation} of \countfunc and \sumfunc aggregate functions are provided (\refsec{sample}).} 

(4) \textbf{Experimental Evaluation.} We conduct extensive experiments on LDBC datasets to evaluate \sys (\refsec{experiment}). The results show that \sys outperforms existing systems by a large margin: compared to Neo4j, \sys achieves 16.21$\times$ average speedup with up to 146.25$\times$ improvement; compared to MySQL, the average speedup reaches 46.67$\times$ with a maximum of 230.53$\times$. Meanwhile, \sys maintains excellent estimation accuracy, with only 0.27\% average error for \countfunc.



\section{Related Work}
\label{sec:related}

\subsection{Data Cube Models}
\label{sec:related:data-cube}

Traditional BI systems employ the multidimensional data cube model~\cite{gray1997data} as their foundational structure. Data is organized using GROUP BY operations to compute numerical metrics across dimensions. The model supports OLAP operations including \drilldownop, \rollupop, and \sliceop, enabling multi-level granularity exploration. Star and snowflake schemas~\cite{kimball2013data,chaudhuri2011overview} serve as the relational foundation for implementing data cubes.

As data relationships become increasingly interconnected, graph-based modeling has gained prominence. Zhao et al.~\cite{zhao2011graph} proposed the graph cube model, extending cube concepts to graph-structured data by aggregating vertices based on attributes while preserving topology. Related work~\cite{gomez2017performing,wang2014pagrol} follows similar aggregation approaches. Shi et al.~\cite{shi2016survey} provided a comprehensive survey of multidimensional analysis on network data.

However, these approaches have two fundamental limitations. First, schemas remain static and require analysts to define the complete structure upfront, lacking flexibility for incremental adjustment during exploration. Second, aggregation is limited to vertex attributes and cannot capture structural semantics like subgraph patterns (e.g., triangles, cliques). In contrast, our hypergraph model supports dynamic schema evolution through operator composition and enables pattern-based retrieval via query graphs.

\subsection{Business Intelligence Tools}
\label{sec:related:bi-tools}

Business Intelligence has evolved over recent decades with numerous tools emerging. Leading platforms include Power BI~\cite{powerbi}, Tableau~\cite{tableau}, QlikSense~\cite{qliksense}, Oracle BI~\cite{oraclebi}, and Sisense~\cite{sisense}. Despite widespread adoption, these traditional tools share fundamental limitations. First, they require analysts to construct wide tables upfront by joining all potentially relevant fact and dimension tables, lacking flexibility to adjust schemas as understanding deepens. Second, these tools cannot automatically synthesize results across multiple analysis rounds, leaving the burden of integrating insights entirely to human analysts.

With Large Language Model (LLM) advancement, new approaches leverage LLM capabilities for BI. Specifically, NL2SQL work~\cite{zhang2024finsql,shen2025study,tang2025llm} translates user questions into SQL queries. NL2BI~\cite{lian2024chatbi} represents a special class handling more complex questions. Lian et al.~\cite{lian2024chatbi} proposed ChatBI for multi-round dialogues, Jiang et al.~\cite{jiang2025siriusbi} proposed SiriusBI with intent querying to understand user intentions, and Weng et al.~\cite{weng2025datalab} proposed DataLab for various BI tasks.

However, these LLM-based approaches face three fundamental limitations. First, their performance still has significant gaps before practical deployment, struggling with complex scenarios. Second, they assume a single clear mapping between user questions and the underlying data schema, bypassing the exploratory process and case-by-case analysis that are essential when ambiguity exists. Third, these AI-driven solutions merely serve as natural language interfaces without addressing system-level challenges like computational cost, lack of reusability, and static schema constraints. In contrast, our \explorativebi framework tackles these system-level limitations through a new data model and operator framework enabling dynamic schema evolution and intermediate result reusability.

\subsection{Approximate Query Processing}
\label{sec:related:aqp}

\revisesigmod{To address computational cost challenges in analytical workloads, sampling-based approximate query processing (AQP) systems have been proposed. BlinkDB~\cite{blinkdb} pre-computes stratified samples of base tables and provides fast approximate answers with bounded errors and response times. VerdictDB~\cite{verdictdb} extends this approach by providing a middleware layer that universalizes AQP across different database backends.}

\revisesigmod{However, these relational AQP systems exhibit fundamental limitations when applied to \explorativebi scenarios. 
First, they typically rely on \textit{offline pre-computed samples} (e.g., stratified samples in BlinkDB, scramble tables in VerdictDB). This approach incurs significant storage and maintenance overhead and lacks the agility to handle ad-hoc exploratory queries where the regions of interest are dynamically determined at runtime.
Second, and more critically, they lack the expressiveness to handle the \textit{structural complexity} of query patterns in \explorativebi. As evidenced by our experiments, relational sampling approaches still suffer from severe performance degradation under multi-way joins, often failing to complete execution within reasonable time limits. Furthermore, these systems may not support complex semantics such as ``non-existence'' in VerdictDB, which is essential for ensuring the induced subgraph constraints. 
In contrast, our hypergraph-based approach performs \textit{on-the-fly sampling} directly over the data structure without pre-materialization, natively supporting multi-way joins and non-existence semantics while providing theoretical guarantees for unbiased estimation.}

\section{Hypergraph Data Model}
\label{sec:model}

Traditional BI systems \revisesigmod{use} star schema and snowflake schema to model multidimensional data using relational tables. However, as relationships between data entities become increasingly interconnected, graph-based data models emerge as more suitable alternatives. Moreover, existing data models struggle to support exploratory scenarios because they are typically static and cannot accommodate dynamic schema evolution.

In this section, we first introduce our hypergraph data model, and then present the operators defined on this model to support \explorativebi.

\subsection{Data Model}
\label{sec:model:data}

We begin by introducing the property graph model, which serves as the foundation for our hypergraph data model.

\begin{definition}[Property Graph]
\label{def:property-graph}
A property graph is defined as $G = (V_G, E_G, \tau_G, 
\mu_{G})$, where:
\begin{itemize}
    \item $V_G$ is the set of vertices,
    \item $E_G \subseteq V_G \times V_G$ is the set of edges. We write an edge as $(v_1, v_2)$, where $v_1, v_2 \in V_G$, and assume undirected edges in this paper for simplicity, 
    meaning that $(v_1, v_2)$ and $(v_2, v_1)$ represent the same edge,
    \item $\tau_G: V_G \cup E_G \to L$ maps vertices and edges to their labels, where $L$ is the set of labels,
    \item $\mu_G: (V_G \cup E_G) \times \prop \to D(\prop)$ maps vertices/edges to their property values, where $D(\prop)$ denotes the domain of property $\prop$.
    \end{itemize}
\end{definition}

In this paper, we use $v.\prop$ for $v \in V_G$ to denote the value of property \prop for vertex $v$, and similarly for edges, and hyperedges (\refdef{hypergraph}). When the context is clear, we write $G = (V, E, \tau, \mu)$ (and even $G = (V, E)$) for short, and this applies to the following definitions of hypergraph (\refdef{hypergraph}) and query graph (\refdef{query-graph}). In this paper, for ease of exposition, we consider property graphs as undirected graphs; however, our definitions can be straightforwardly extended to directed graphs.

Building upon the property graph model, we can perform multidimensional analysis based on vertex and edge attributes, similar to traditional data cube models. However, this approach only supports classification based on attribute values. In real-world scenarios, we often need to classify and aggregate not just individual vertices and edges, but entire subgraphs based on their structural properties~\cite{zhao2011graph}. For instance, computing the count of all mutually connected groups of three people requires first identifying all triangular patterns among Person vertices based on connectivity, then performing aggregation. To address this limitation, we first give the definition of induced subgraphs and hyperedges, then propose the hypergraph data model.

\begin{figure*}
    \centering
    \begin{subfigure}[b]{0.36\textwidth}
        \includegraphics[width=\textwidth]{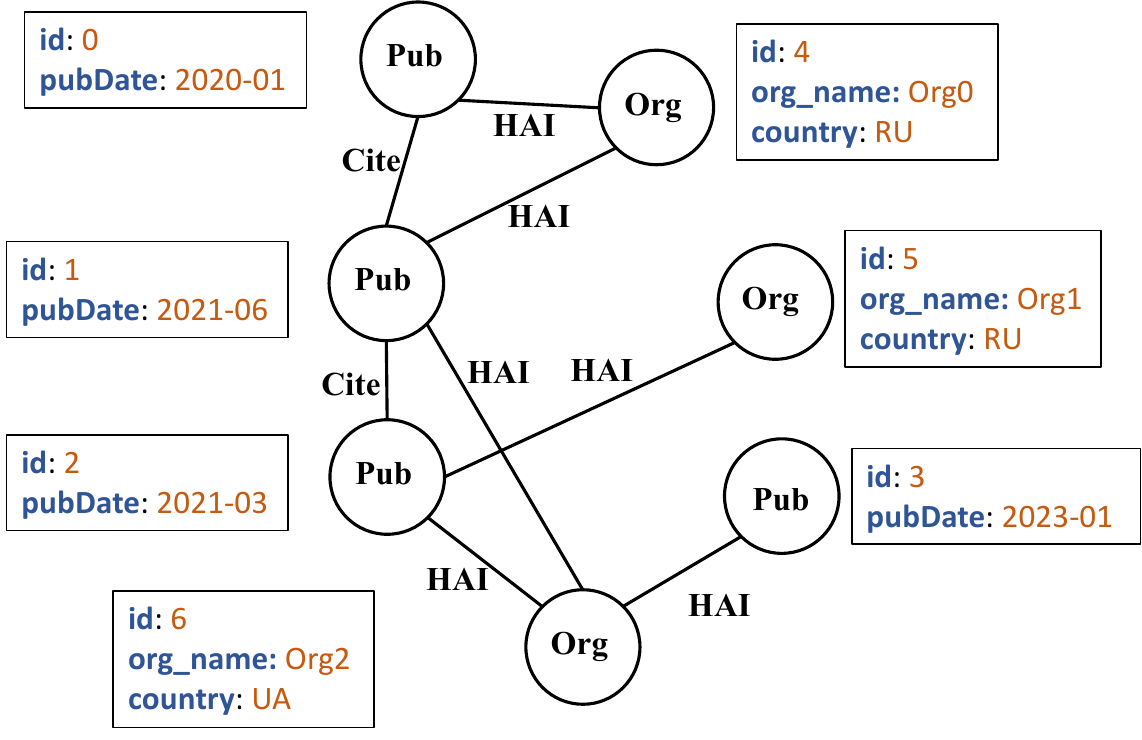}
        \caption{Property Graph $G_0$}
        \label{fig:example:schema:property-graph}
    \end{subfigure}
    \hfill
    \begin{subfigure}[b]{0.23\textwidth}
        \centering
        \includegraphics[width=.6\textwidth]{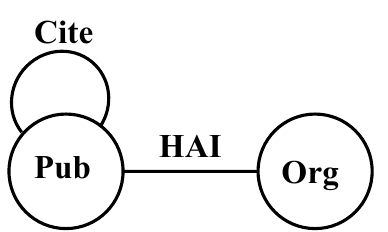}
        \caption{Property Graph Schema}
        \label{fig:example:schema:property-graph-schema}
        \vspace{1em}
        
        \includegraphics[width=\textwidth]{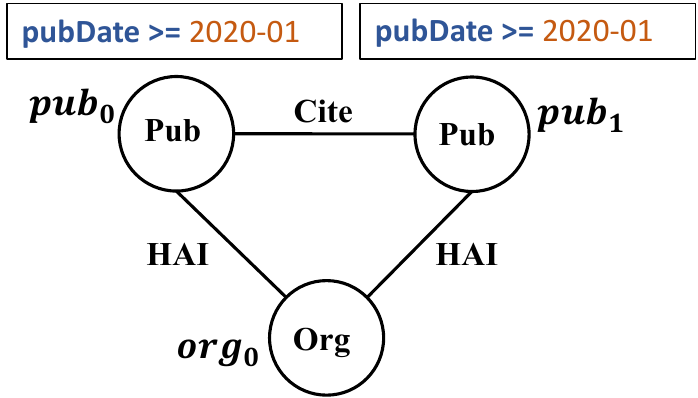}
        \caption{Query Graph $q_{\text{s}}$}
        \label{fig:example:person-triangle-query}
    \end{subfigure}
    \hfill
    \begin{subfigure}[b]{0.36\textwidth}
        \includegraphics[width=\textwidth]{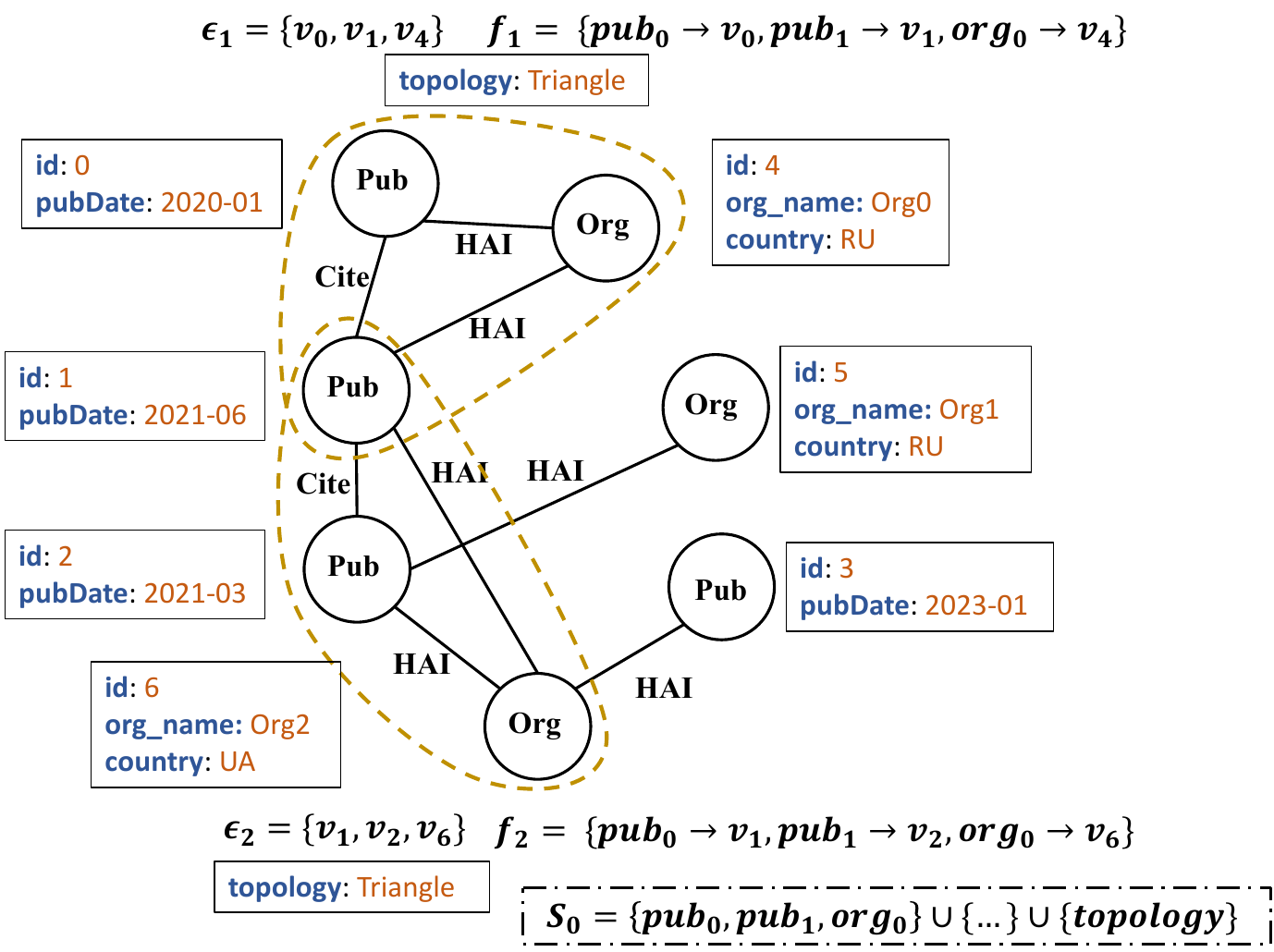}
        \caption{Hypergraph $\mathcal{G}_0\left<G_0\right>$}
        \label{fig:example:schema:hypergraph}
    \end{subfigure}
    \caption{An example of a property graph, property graph schema, query graph, and the corresponding hypergraph. \textit{Pub}, \textit{Org}, and \textit{HAI} are abbreviations of \textit{Publication}, \textit{Organization}, and \textit{hasAuthorInstitution}, respectively. $v_i$ represents the vertex with id $i$; $u_i$ represents query vertex with id $i$.}
    \label{fig:example:person-project-example}
\end{figure*}

\begin{definition}[Induced Subgraph]
    \label{def:induced-subgraph}
    Given a property graph $G = (V, E)$, $G_{\text{sub}} = (V_s, E_s)$ is an induced subgraph of $G$ if and only if: (1) $V_s \subseteq V$ and $E_s \subseteq E$; (2) $\forall (u, v) \in E$, if $u \in V_s$ and $v \in V_s$, then $(u, v) \in E_s$.
\end{definition}

\comment{
\begin{definition}[Hyperedge]
\label{def:hyperedge}
Given a property graph $G = (V, E, \tau,$ $\kappa, \mu)$, a hyperedge $\epsilon$ contains a set of vertices $V_{\epsilon} \subseteq V$.
\end{definition}
}

\begin{definition}[Hypergraph Data Model]
\label{def:hypergraph}
Given a \\ property graph $G = (V_G, E_G, \tau_G, \mu_G)$, we denote $\mathcal{G}\left<G\right> = (\mathcal{V}_{\mathcal{G}}, \mathcal{E}_{\mathcal{G}}, 
 \mu_{\mathcal{G}})$ as a hypergraph regarding $G$, where:
\begin{itemize}
    \item $\mathcal{V}_{\mathcal{G}} = V_G$,
    \item $\mathcal{E}_{\mathcal{G}}$ is the set of hyperedges, where each hyperedge encode a set of graph vertices $V_\epsilon$, denoted as $V_{\epsilon} \subseteq \mathcal{V}_G$,
    \item $\mu_{\mathcal{G}}: \mathcal{E}_{\mathcal{G}} \times \prop \to D(\prop)$ maps hyperedges to their property values. 
\end{itemize}
\end{definition}

When the underlying property graph $G$ is not of concern, we abbreviate the hypergraph as $\mathcal{G}$. It is important to note that a hyperedge is not merely a set of vertices—given the underlying graph $G$, the vertices in a hyperedge form an induced subgraph of $G$, which can exhibit structural characteristics. The properties of hyperedges, defined through $\mu_{\mathcal{G}}$, can capture such structural information, as we will demonstrate when introducing the hypergraph schema (\refdef{hypergraph-schema}).

\comment{
To keep the notation concise, we model these structural characteristics as built-in properties of hyperedges through $\mu_{\mathcal{G}}$. One such property is \textbf{connectivity pattern} that captures the topological structure among hyperedge vertices in the underlying graph. For example, if there are three vertices, they can form 8 types of structures including triangle, 2-path, and three isolated points.

Under this hypergraph data model definition, we can classify hyperedges based on both vertex/edge attributes from $G$ and structural properties from $\mu_{\mathcal{G}}$, enabling aggregation based on both attribute and structural information.
}

\comment{
\begin{remark}
    In this paper, we focus on scenarios where the vertices in each hyperedge form an induced subgraph in the property graph that satisfies user requirements. That is, each hyperedge corresponds to an induced subgraph.
    Theoretically, each hyperedge can be an arbitrary set of vertices, corresponding to any subgraph formed by these vertices in the property graph. For example, the vertices in a hyperedge together with some edges could form a connected component or a minimum spanning tree. We will explore these more general scenarios in future work.
\end{remark}
}

\begin{example}
    \label{ex:hypergraph-construction}
    Consider a property graph containing two types of vertices: Publication (abbr.~Pub) and Organization (abbr.~Org), as shown in \reffig{example:schema:property-graph-schema}. Pub vertices can be connected by Cite edges, while Pub and Org vertices can be connected by hasAuthorInstitution edges (abbr.~HAI). \reffig{example:schema:property-graph} shows a property graph instance with four Pub vertices and three Org vertices. \reffig{example:schema:hypergraph} shows the corresponding hypergraph with two hyperedges $\epsilon_1$ and $\epsilon_2$. Take hyperedge $\epsilon_1$ as an example: it contains vertices $v_0, v_1,$ and $v_4$, which form an induced subgraph in the underlying graph.
\end{example}


\subsection{Operators on Hypergraph Data Model}
\label{sec:model:operators}

We first present three core operators based on the hypergraph data model, namely, \sourceop, \joinop, and \viewop, and then discuss the adaptations of traditional cube operators to hypergraphs.
These operators exhibit a closure property: \sourceop serves as the entry point (Graph $\to$ Hypergraph), \viewop as the exit point (Hypergraph $\to$ Table), while all others (\joinop, \drilldownop, \rollupop, \sliceop, \diceop) form a closed algebra (Hypergraph $\to$ Hypergraph), enabling composable chains for iterative exploration.

\subsubsection{\sourceop Operator}
\label{sec:model:source}

Before defining the \sourceop operator, we introduce the concept of induced subgraph isomorphism.

\begin{definition}[Query Graph]
\label{def:query-graph}
A query graph is defined as $q = (V_q, E_q, \tau_q)$, where $V_q$ is the set of vertices, $E_q$ is the set of edges, and $\tau_q$ maps vertices and edges to their labels.
\end{definition}

The query graph serves as a template for subgraph matching. Optionally, users can specify constraints on vertex and edge property values to filter matching results. \reffig{example:person-triangle-query} shows an example of a query graph with three vertices $pub_0, pub_1,org_0$. \revise{The two Pub vertices are required to have publication dates after January, 2020}.

To perform pattern matching with query graphs, we adopt the induced subgraph isomorphism semantics following Carletti et al.~\cite{carletti2017challenging}. This ensures that matching results precisely capture the connectivity structure specified in the query graph.

\begin{definition}[Induced Subgraph Isomorphism]
\label{def:induced-subgraph-isomorphism}
~\\
Given a property graph $G = (V, E, \tau, \mu)$ and a query graph $q = (V_q, E_q, \tau_q)$, a subgraph $G_s = (V_s, E_s) \subseteq G$ is an induced subgraph isomorphic to $q$ if there exists a bijection $f: V_q \to V_s$ such that:
\begin{itemize}
    \item \textbf{Label preservation}: $\forall u \in V_q$, $\tau_q(u) = \tau(f(u))$,
    \item \textbf{Edge preservation}: $\forall (u_1, u_2) \in E_q$, $(f(u_1), f(u_2)) \\ \in E_s$ and $\tau_q((u_1,u_2)) = \tau((f(u_1), f(u_2)))$,
    \item \textbf{Induced structure}: $G_s$ is an induced subgraph of $G$, and $\forall (v_1, v_2)$ $\in E_s$, $(f^{-1}(v_1), f^{-1}(v_2)) \in E_q$ and $\tau((v_1,v_2)) = \tau_q((f^{-1}(v_1),$ $f^{-1}(v_2)))$.
\end{itemize}
The bijection $f$ is called a matching, or matched instance of $q$ in $G$. If property constraints are specified in $q$, the matching must also satisfy these constraints. 
\end{definition}

Note that we do not require the query graph to be connected. If $q$ is disconnected, the induced subgraph constraint ensures that matched subgraphs must also be disconnected. This design choice allows analysts to explicitly encode disconnectivity as an important structural constraint in exploratory analysis.

\begin{example}
    \label{ex:induced-subgraph-isomorphism}
    Given the property graph $G_0$ in \reffig{example:schema:property-graph}, suppose we want to find organizations with two cited publications \revise{published after January, 2020}. We define query graph $q_s$ (\reffig{example:person-triangle-query}) with query vertices $pub_0, pub_1$ (Pub) and $org_0$ (Org). The matching results are shown in \reffig{example:schema:hypergraph} with dashed lines: $f_1 = \{pub_0 \to v_0, pub_1 \to v_1, org_0 \to v_4\}$ and $f_2 = \{pub_0 \to v_1, pub_1 \to v_2, org_0 \to v_6\}$. Note that symmetric matchings (e.g., $f'_1 = \{pub_0 \to v_1, pub_1 \to v_0, org_0 \to v_4\}$) also exist but are omitted for simplicity as they induce the same vertex sets.
\end{example}


Now we can define the \sourceop operator based on induced subgraph isomorphism.

\begin{definition}[Source Operator, $\source$]
\label{def:source-operator}
$\source: G \\ \times q \to \mathcal{G}\left<G\right>$. Given a property graph $G = (V, E)$ and a query graph $q = (V_q, E_q)$, the \sourceop operator finds all induced subgraphs in $G$ that are isomorphic to $q$ (\refdef{induced-subgraph-isomorphism}). The operator returns a hypergraph $\mathcal{G}\left<G\right> = (\mathcal{V}_{\mathcal{G}}, \mathcal{E}_{\mathcal{G}}, \mu_{\mathcal{G}})$ where each hyperedge $\epsilon \in \mathcal{E}_{\mathcal{G}}$ corresponds to a matched induced subgraph $G_s = (V_s, E_s)$ with $\epsilon = V_s$.

For each hyperedge $\epsilon$, the property mapping function $\mu_{\mathcal{G}}$ includes a built-in property ``\textit{topology}'' that captures the connectivity pattern of the induced subgraph $G_s$. Specifically, $\mu_{\mathcal{G}}(\epsilon, \text{topology})$ (or $\epsilon.\text{topology}$) is determined by the edge set $E_s$ of the induced subgraph: two hyperedges have the same topology if and only if their induced subgraphs are isomorphic as graphs (i.e., same connectivity structure).
\end{definition}

\begin{definition}[Hypergraph Schema]
\label{def:hypergraph-schema}
Given a hypergraph $\mathcal{G}\left<G\right>$ regarding a query graph $q = (V_q, E_q, \tau_q)$, its schema $\mathcal{S}_q$ is defined as:
$$\mathcal{S}_q = V_q \cup E_q \cup \mathcal{P}_{\mathcal{G}}$$
where $V_q$ represents query vertices (vertex columns), $E_q$ represents query edges (edge columns), and $\mathcal{P}_{\mathcal{G}}$ is the set of hyperedge property names (hyperedge columns).

For each hyperedge $\epsilon$ in the hypergraph, there exists a corresponding matching $f: V_q \to V_\epsilon$ (i.e., a matched instance as defined in \refdef{induced-subgraph-isomorphism}). The schema defines the following three types of columns based on this matching:
\begin{itemize}
    \item \textbf{Vertex columns}: For each query vertex $u \in V_q$, there exists a unique vertex $v = f(u) \in V_\epsilon$ in the underlying graph $G$. We can access properties of this vertex using the notation $u.\prop$, which denotes the value of $\prop$ for vertex $f(u)$ in $G$. 
    \item \textbf{Edge columns}: For each query edge $(u_1, u_2) \in E_q$, the corresponding edge $(f(u_1), f(u_2))$ exists in the induced subgraph. We can access its properties using $(u_1, u_2).\prop$, which denotes the value of $\prop$ for edge $(f(u_1), f(u_2))$ in $G$.
    \item \textbf{Hyperedge columns}: For each hyperedge property name $\prop \in \mathcal{P}_{\mathcal{G}}$, we can access the property value using $\epsilon.\prop$. These properties are defined through $\mu_{\mathcal{G}}$ and include structural properties like ``\text{topology}''.
\end{itemize}
\end{definition}

The hypergraph schema represents the structural and property information maintained throughout hypergraph data model operations. The association between a hypergraph and its query graph can be established in two ways: (1) explicitly through the \sourceop operator (\refdef{source-operator}), where pattern matching against query graph $q$ generates hyperedges with their matchings; (2) implicitly through the \joinop operator (\refdef{join-operator}), which composes two query graphs to form an extended query graph while maintaining schema consistency. Throughout all operations, the schema preserves its three-part structure of vertex, edge, and hyperedge columns.

The hypergraph schema provides a uniform view over all hyperedges: each hyperedge can be conceptually viewed as a row in a relational table, where columns correspond to query vertices, query edges, and hyperedge properties. Cell values for vertex and edge columns reference the actual vertices and edges in the underlying property graph $G$ through the matching $f$, while hyperedge columns contain structural and other properties of the subgraphs themselves. This schema-based view naturally enables access to both graph properties from $G$ and structural properties of hyperedges, and serves as the foundation for bridging hypergraphs and traditional data cube models (see \refsec{model:view} and \refsec{model:cube-operators}).

\begin{remark}[Structural Properties of Hyperedges]
\label{remark:structural-properties}
The hyperedge properties in $\mathcal{P}_{\mathcal{G}}$ can capture structural characteristics of the induced subgraphs. A key example is the built-in \text{topology} property (introduced in \refdef{source-operator}), which describes the connectivity pattern among vertices in a hyperedge. For instance, consider a hyperedge containing three vertices induced from a query graph with three edges. The induced subgraph can form different topologies depending on which edges are actually present in $G$:
\begin{itemize}
    \item \textbf{Triangle}: all three vertices are mutually connected (all three edges present)
    \item \textbf{2-path}: vertices form a linear chain (exactly two edges present)
    \item \textbf{Disconnected}: fewer than two edges present, or vertices are completely isolated
\end{itemize}
Such topology can be efficiently encoded as a bitset where each bit represents edge presence (1 for connected, 0 for disconnected) between vertex pairs. This encoding enables efficient classification and aggregation based on graph patterns, which is essential for \explorativebi analysis. More generally, we can consider incorporating other graph analytics (e.g., community detection) to define the \sourceop, as well as additional structural properties such as modularity score and density. We will explore these possibilities in future work.
\end{remark}

\begin{example}
\label{ex:source-operator}
\revisesigmod{Continuing from the previous example, applying $\source(G_0, q_s)$ returns a hypergraph $\mathcal{G}_{0}\left<G_0\right>$ with two hyperedges $\epsilon_1$ and $\epsilon_2$ corresponding to matchings $f_1$ and $f_2$ (\reffig{example:schema:hypergraph})}. The hypergraph schema $\mathcal{S}_{q_s}$ $=$ $\{pub_0,$ $pub_1,$ $org_0\}$ $\cup E_{q_s} \cup \{\text{topology}\}$ enables property access via query vertices, e.g., $pub_0.\textit{id}$ evaluates to $v_0.\textit{id}$ for $\epsilon_1$ where $f_1(pub_0) = v_0$. Both hyperedges have $\epsilon.\text{topology} = \text{Triangle}$ since all three vertices are mutually connected in their induced subgraphs.
\end{example}

\subsubsection{\joinop Operator}
\label{sec:model:join}

To support exploratory analysis, the data model must accommodate dynamic schema evolution. The \joinop operator enables this by merging two hypergraphs to incorporate new data dimensions as analytical requirements deepen.


\begin{figure*}[ht]
    \centering
    \begin{subfigure}[b]{0.48\textwidth}
        \centering
        \includegraphics[width=.5\textwidth]{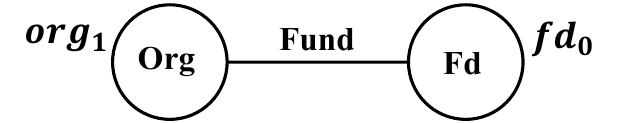}
        \caption{Query graph $q_{\text{f}}$ for the \joinop operator}
        \label{fig:example:org-funding-query}
        \vspace{1em}
        
        \includegraphics[width=.82\textwidth]{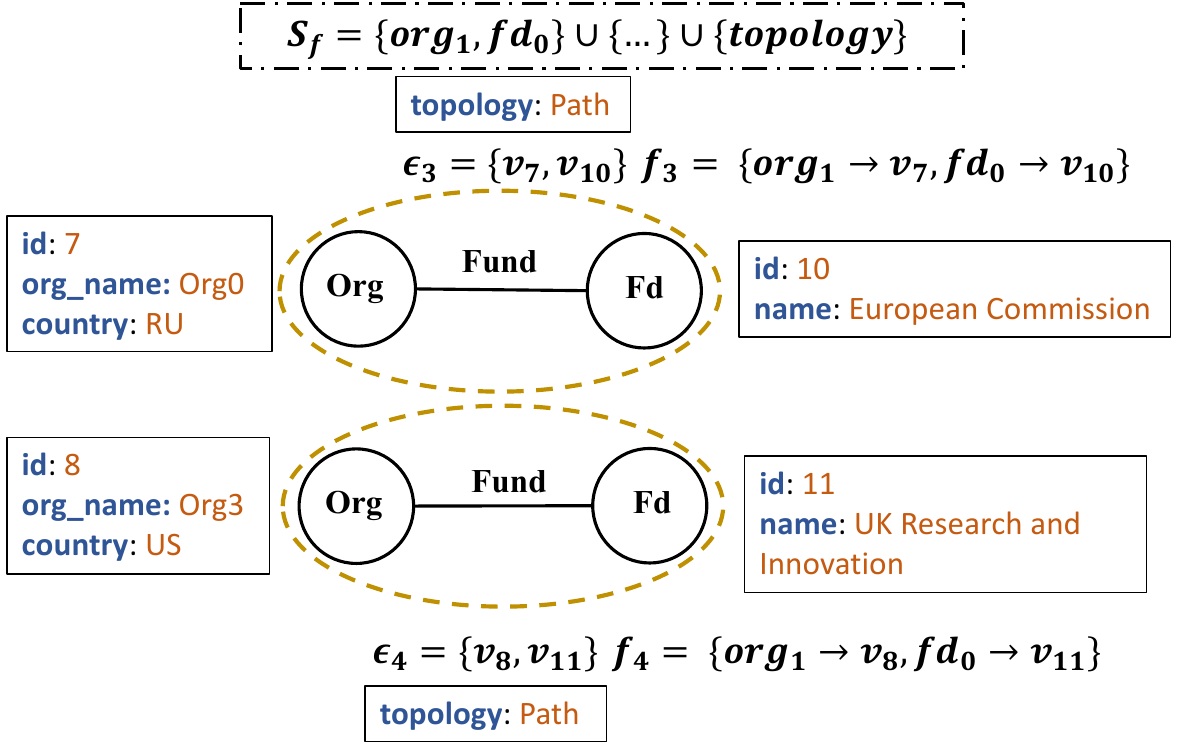}
        \caption{Hypergraph $\mathcal{G}_{\text{f}}$. Fd is the short name of Funding.}
        \label{fig:example:org-funding-hyper}
    \end{subfigure}
    \hfill
    \begin{subfigure}[b]{0.48\textwidth}
        \centering
        \includegraphics[width=0.88\textwidth]{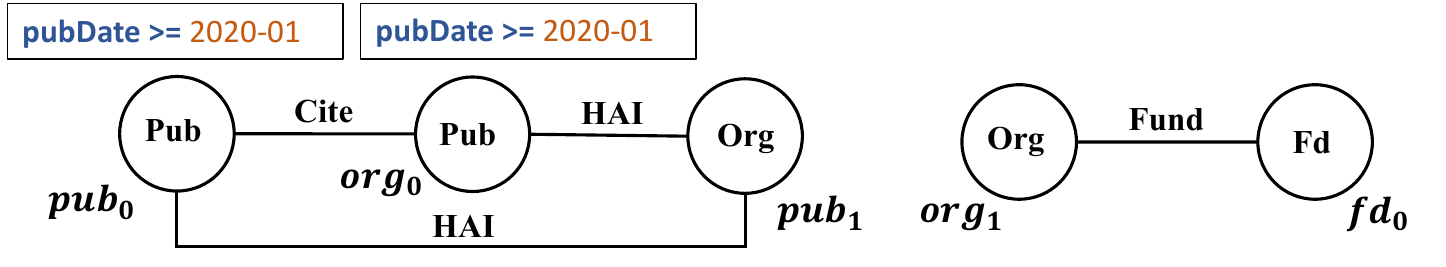}
        \caption{The query graph for $\mathcal{G}_{\text{join}}$}
        \label{fig:example:join-query}
        \vspace{1em}

        \includegraphics[width=.92\textwidth]{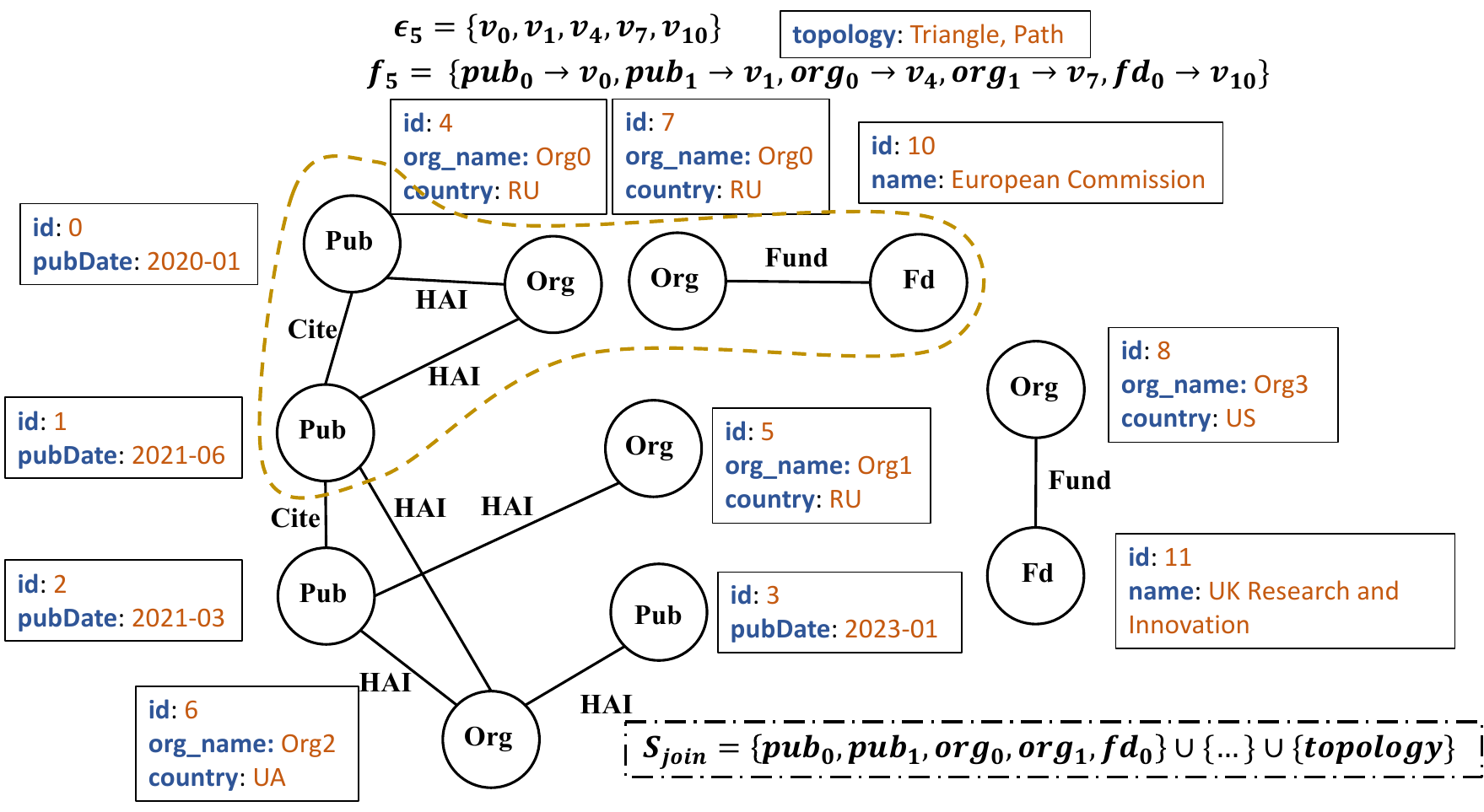}
        \caption{Hypergraph $\mathcal{G}_{\text{join}}$ obtained by joining $\mathcal{G}_0$ and $\mathcal{G}_{\text{f}}$ with $\sigma: \{(org_0.\text{org\_name}, org_1.\text{org\_name})\}$}
        \label{fig:example:pub-org-funding-hyper}
    \end{subfigure}
    \caption{An example of the Join operator}
    \label{fig:example:person-place-join}
\end{figure*}

\begin{definition}[Join Condition, $\sigma$]
\label{def:join-condition}
\revisesigmod{Let hypergraphs \\ $\mathcal{G}$$\left<G\right>$ and $\mathcal{G}'$$\left<G'\right>$ correspond to query graphs $q = (V_q, E_q,$ $\tau_q)$ and $q' = (V_{q'}, E_{q'}, \tau_{q'})$ respectively, where $\mathcal{S}_q$ and $\mathcal{S}_{q'}$ are the corresponding hypergraph schemas (\refdef{hypergraph-schema}).
A join condition is defined as a set of vertex property equality constraints}:
$$\sigma = \{(u_1.\prop_1, u'_1.\prop'_1), \cdots, (u_k.\prop_k, u'_k.\prop'_k)\}$$
where $u_i \in V_q$ and $u'_i \in V_{q'}$ for $i = 1,\cdots,k$. 

Two hyperedges $\epsilon \in \mathcal{E}_{\mathcal{G}}$ and $\epsilon' \in \mathcal{E}_{\mathcal{G}'}$ with their respective matchings $f: V_q \to V_\epsilon$ and $f': V_{q'} \to V_{\epsilon'}$ satisfy the join condition $\sigma$ if and only if:
$$\forall (u_i.\prop_i, u'_i.\prop'_i) \in \sigma \Rightarrow f(u_i).\prop_i = f'(u'_i).\prop'_i.$$
\end{definition}

\begin{definition}[\joinop Operator, $\join$]
\label{def:join-operator}
Given hypergraphs $\mathcal{G}\left<G\right>$ and $\mathcal{G}'\left<G'\right>$ with their associated components as defined in \refdef{join-condition}, and a join condition $\sigma$, the \joinop operator is defined as:
$$\join: \mathcal{G}\left<G\right> \times \mathcal{G}'\left<G'\right> \times \sigma \to \mathcal{G}_r\left<G \cup G'\right>$$

The operator returns a hypergraph $\mathcal{G}_r\left<G \cup G'\right> = (\mathcal{V}_{\mathcal{G}_r}, \mathcal{E}_{\mathcal{G}_r}, \mu_{\mathcal{G}_r})$ regarding an implicitly composed query graph $q_r$, where:
\begin{itemize}
    \item \textbf{Underlying property graph}: $G \cup G'$ combines vertices and edges from both $G$ and $G'$.
    \item \textbf{Composed query graph}: $q_r = (V_{q_r}, E_{q_r}, \tau_{q_r})$ where $V_{q_r} = V_q \cup V_{q'}$ and $E_{q_r} = E_q \cup E_{q'}$. This composition implicitly extends the query graph structure while maintaining consistency with the original query graphs.
    \item \textbf{Extended schema}: $\mathcal{S}_{q_r} = \mathcal{S}_q \cup \mathcal{S}_{q'} = (V_q \cup V_{q'}) \cup (E_q \cup E_{q'}) \cup (\mathcal{P}_{\mathcal{G}} \cup \mathcal{P}_{\mathcal{G}'})$, preserving the three-part structure of vertex columns, edge columns, and hyperedge columns.
    \item \textbf{Hyperedges}: For each pair of hyperedges $\epsilon \in \mathcal{E}_{\mathcal{G}}$ and $\epsilon' \in \mathcal{E}_{\mathcal{G}'}$ with matchings $f$ and $f'$ respectively that satisfy the join condition $\sigma$ (\refdef{join-condition}), a new hyperedge $\epsilon_r$ is created in $\mathcal{E}_{\mathcal{G}_r}$ where: (1) For $\epsilon_r$, we have $V_{\epsilon_r} = V_\epsilon \cup V_{\epsilon'}$; (2) The matching $f_r: V_{q_r} \to V_{\epsilon_r}$ is defined as: $f_r(u) = f(u)$ for $u \in V_q$ and $f_r(u') = f'(u')$ for $u' \in V_{q'}$.
\end{itemize}
\end{definition}

\begin{example}
\label{ex:join-operator}
Given $\mathcal{G}_0$ as shown in \reffig{example:schema:hypergraph}, suppose we want to add the Funding dimension for analysis. We first construct a hypergraph $\mathcal{G}_{\text{f}}$ as shown in \reffig{example:org-funding-hyper} with the query graph $q_{\text{f}}$ in \reffig{example:org-funding-query}.
In $\mathcal{G}_{\text{f}}$, each hyperedge represents an Organization-Funding relationship. 
We then apply the \joinop operator to merge $\mathcal{G}_0$ and $\mathcal{G}_{\text{f}}$ with join condition $\sigma = \{(org_0.\textit{org\_name}, org_1.\textit{org\_name})\}$, where $org_0 \in V_{q_s}$ and $org_1 \in V_{q_{\text{f}}}$.
The resulting hypergraph $\mathcal{G}_{\text{join}}$ is shown in \reffig{example:pub-org-funding-hyper}.
\revisesigmod{Specifically, $\epsilon_1$ and $\epsilon_3$ satisfy the join condition $\sigma$ because $f_1(org_0) = v_4$, $f_3(org_1) = v_7$, and $v_4.\textit{org\_name} = v_7.\textit{org\_name}$\text{=}$\text{``Org0''}$.}
Therefore, they are merged into hyperedge $\epsilon_5$ in $\mathcal{G}_{\text{join}}$. In contrast, $\epsilon_2$ has no matching hyperedge in $\mathcal{G}_{\text{f}}$ that satisfies $\sigma$, so $\mathcal{G}_{\text{join}}$ contains only one hyperedge.
The query graph for $\mathcal{G}_{\text{join}}$ is shown in \reffig{example:join-query}, and $\mathcal{S}_{\text{join}} = \mathcal{S}_{q_s} \cup \mathcal{S}_{q_{\text{f}}}$.
\end{example}

\begin{remark}[Design Choice: Separate Underlying Graphs]
\label{remark:separate-graphs}
In the \joinop operator, we explicitly define the two hypergraphs to have different underlying property graphs $G$ and $G'$. This design avoids complications when vertices in $V_\epsilon$ and $V_{\epsilon'}$ overlap by treating $G$ and $G'$ as initially disjoint. Importantly, this preserves the isomorphism semantics while ``matching'' the composed query graph $q_r$, as well as maintaining the set semantics of the resulting hyperedge. When both hypergraphs originate from the same graph (e.g., via different \sourceop operations), we can \emph{conceptually} view $G$ and $G'$ as two copies, but they remain the same graph in implementation.
\end{remark}

\subsubsection{\viewop Operator}
\label{sec:model:view}

While the hypergraph structure provides flexibility for modeling complex and dynamic schema in \kw{Exploratory} \kw{BI}, it can be challenging to directly analyze and interpret. To leverage the existing BI ecosystem—including traditional BI techniques, aggregate functions, and visualization toolkits—we introduce the \viewop operator that projects hypergraphs into relational tables, enabling seamless integration with established BI tools.

\begin{definition}[\viewop Operator, $\view$]
\label{def:view-operator}
Given a hypergraph $\mathcal{G}\left<G\right> = (\mathcal{V}_{\mathcal{G}}, \mathcal{E}_{\mathcal{G}}, \mu_{\mathcal{G}})$ with query graph $q = (V_q, E_q,$ $\tau_q)$ and schema $\mathcal{S}_q = V_q \cup E_q \cup \mathcal{P}_{\mathcal{G}}$, the operator $\view: \mathcal{G}\left<G\right> \to R$ projects $\mathcal{G}\left<G\right>$ into a relational table $R$ where each hyperedge $\epsilon$ with matching $f$ corresponds to one row. The columns are derived from vertex, edge, and hyperedge properties as defined in $\mathcal{S}_q$.
\end{definition}


In general, when hyperedges in a hypergraph are generated from different query graphs, projecting them into a single relational table requires constructing a unified schema that accommodates all vertex and edge properties, with NULL values for missing columns. In this paper, we simplify by assuming all hyperedges in a hypergraph are generated from the same query graph (ensured by our \sourceop and \joinop operators). However, hyperedges may still contain different numbers of edges depending on the topology of their induced subgraphs. For example, given a query graph with three vertices and three edges, one hyperedge might form a triangle (all three edges present), while another forms a path (only two edges present). We can handle such varying edge structures by including all possible edges from the query graph in the schema and setting absent edges to NULL.


\begin{figure*}
    \centering
    \includegraphics[width=0.95\textwidth]{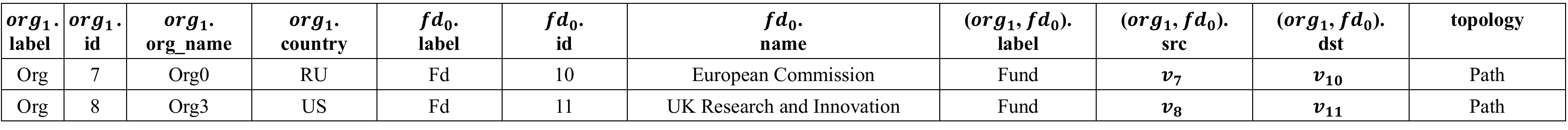}
    \caption{The relational table obtained by applying the \viewop operator to the hypergraph $\mathcal{G}_{\text{f}}$}
    \label{fig:example:org-funding-view}
\end{figure*}

\begin{example}
\label{ex:view-operator}
Applying the \viewop operator to $\mathcal{G}_{\text{f}}$ (\reffig{example:org-funding-hyper}) produces table $R$ (\reffig{example:org-funding-view}), where each row corresponds to a hyperedge in $\mathcal{G}_{\text{f}}$.
\end{example}

\subsubsection{Adaptation to Cube Operators}
\label{sec:model:cube-operators}

To support multi-dimensional analysis on hypergraphs, we can conceptually turn the hypergraph schema into a relational view as in the \viewop operator, apply traditional BI cube operators~\cite{gray1997data} such as \drilldownop, \rollupop, \sliceop, and \diceop, and then transform back to the hypergraph representation—all without changing the physical data layout. These operators do not alter the underlying query graph but rather operate at the schema level by manipulating which properties are used for grouping and aggregation.

For example, \drilldownop refines the analysis by increasing granularity—moving from coarser to finer levels of detail (e.g., drilling down to specific vertex properties and even topologies). Conversely, \rollupop reduces granularity by aggregating hyperedges at higher levels. \sliceop filters hyperedges by fixing specific property values on one dimension, while \diceop selects hyperedges satisfying constraints across multiple dimensions. Since all hyperedges in a hypergraph share the same query graph structure (ensured by our \sourceop and \joinop operators), they naturally project to a uniform relational schema, enabling seamless application of these cube operators. Therefore, we omit their formal definitions in this paper.

\section{Framework Overview}
\label{sec:overview}

To support \explorativebi and address the challenges discussed in Section~\ref{sec:intro}, we design and implement the \sys system. As illustrated in \reffig{overview:framework}, the \sys framework consists of three main layers: Storage Layer, Computational Layer, and Analysis Layer. In the following subsections, we detail the functionality and design of each layer.

\begin{figure}
    \centering
    \includegraphics[width=\linewidth]{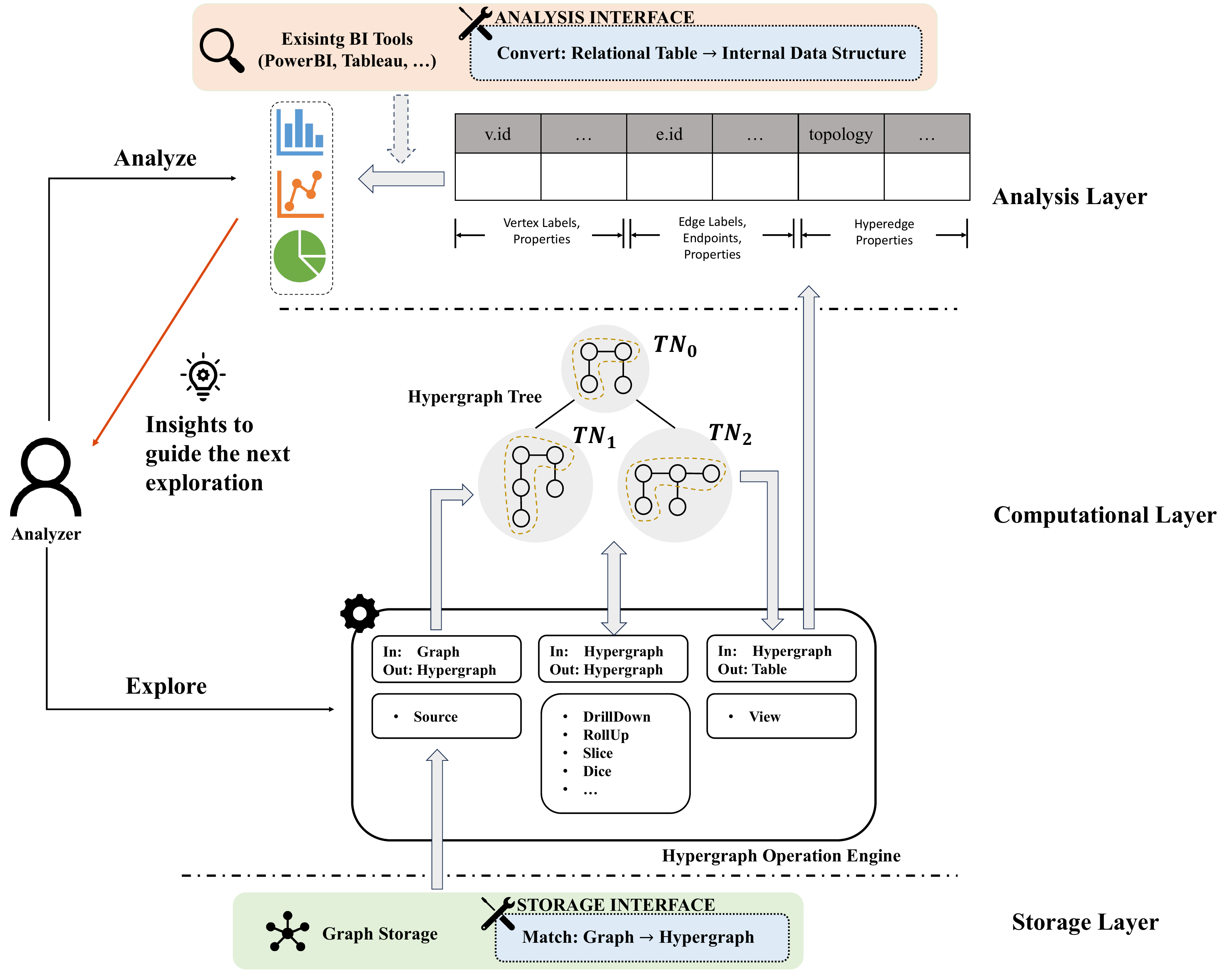}
    \caption{The architecture of the \sys framework}
    \label{fig:overview:framework}
\end{figure}

\subsection{Storage Layer}
\label{sec:overview:storage}

\revisesigmod{The Storage Layer is responsible for data persistence and management. 
Although our Hypergraph Data Model is logically built on a property graph abstraction (see \refsec{model}), the underlying storage architecture is designed to be agnostic. In principle, any database system capable of supporting induced subgraph matching --- whether a native graph database (e.g., Neo4j~\cite{neo4j}) or a relational database (via schema mapping~\cite{lou2024towards}) --- can serve as the backend. }

\revisesigmod{However, as highlighted in \refsec{intro}, generic database engines currently fall short of the performance requirements for \explorativebi on large-scale datasets. Relational engines incur high costs for complex join-based matching, while native graph databases lack the specialized sampling primitives needed for efficient estimation. Consequently, \sys currently implements a custom storage layer. We build this layer by extending the DAF algorithm~\cite{han2019efficient} to rigorously support induced subgraph matching~\cite{carletti2017challenging} and integrating it directly with our sampling-based optimization engine. This implementation serves as a performance enabler, allowing \sys to meet interactive analysis demands while retaining the architectural flexibility to support standard databases in the future.}

\subsection{Computational Layer}
\label{sec:overview:computational}

The Computational Layer serves as the core of \sys, managing the interactions with both the Storage Layer and the Analysis Layer while orchestrating hypergraph operations to enable exploratory analysis. As shown in \reffig{overview:framework}, this layer comprises two key components: the Hypergraph Operation Engine and Hypergraph Tree.

\stitle{Hypergraph Operation Engine.} The Hypergraph Operation Engine implements all operators defined on the hypergraph data model, e.g., \sourceop, \joinop, and \viewop etc. These operators enable dynamic schema evolution and iterative analysis by transforming hypergraphs to incorporate new dimensions or refine analytical focus.

Among these operators, the \sourceop operator plays a unique role. While most operators take hypergraphs as both input and output, the \sourceop operator takes the graph database as input and produces a hypergraph as output. It constructs hypergraphs from the underlying property graph by performing subgraph matching based on query patterns (as defined in \refdef{source-operator}). Therefore, the \sourceop operator serves as the bridge connecting the Storage Layer and the Computational Layer, enabling data retrieval and hypergraph construction. Additionally, the \viewop operator bridges the Computational Layer and the BI Analysis Layer by taking a hypergraph as input and producing a relational table as output, which is then passed to the Analysis Layer for analysis.

\stitle{Hypergraph Tree.} The Hypergraph Tree maintains a record of hypergraphs generated during the exploratory analysis process. Each node in the tree represents a hypergraph, and child nodes are derived from their parent nodes through the application of operators from the Hypergraph Operation Engine. For example, in \reffig{overview:framework}, the hypergraph of $TN_1$ is obtained by applying the \joinop operator to the hypergraph of its parent node $TN_0$, incorporating a new data dimension for analysis.

This tree structure captures the evolution of analytical understanding throughout the exploration process, enabling analysts to review previous states, compare analytical paths, and reuse intermediate results. Analysts can select hypergraphs from the tree, apply operators to generate new hypergraphs, and pass them to the Analysis Layer for examination.

\subsection{Analysis Layer}
\label{sec:overview:analysis}

The Analysis Layer receives the relational table from the Computational Layer, which is obtained with the \viewop operator, and then performs aggregation and analysis on it. This transformation enables the application of traditional BI aggregation operations and facilitates result interpretation. Specifically, analysts select specific dimensions (columns) for grouping and apply aggregate functions to compute measures of interest. 
The aggregate functions currently supported in our system include \countfunc, \distinctcountfunc, \maxfunc, \minfunc, and \sumfunc. These aggregation functions cover typical BI analysis scenarios \cite{xu2023efficiently,yi2014indexing}. In the future, we will consider supporting additional aggregate functions.

We have designed an analysis interface to integrate the analytical capabilities of existing BI tools such as Power BI~\cite{powerbi} and Tableau~\cite{tableau}. Specifically, if the convert function in the interface is implemented in an existing BI tool, the relational table transformed by the \viewop operator can be converted into the internal data structure of that BI tool, thereby enabling the use of existing BI tools to perform subsequent analytical tasks on the relational table.

Importantly, insights derived from the Analysis Layer can guide subsequent operations in the Computational Layer. Based on analytical results, analysts may identify the need to incorporate additional data dimensions (through \joinop operations), refine the data scope (through \sliceop operations), or adjust the granularity of analysis (through \drilldownop or \rollupop operations). This feedback loop between the Analysis Layer and the Computational Layer embodies the essence of exploratory analysis, enabling analysts to iteratively refine their understanding and progressively discover insights.

\comment{
\revise{
\begin{remark}[Agentic Explorable BI]
The structured \explorativebi workflow naturally suits agent-driven automation. Agents can invoke operators (\sourceop, \joinop, \viewop) to construct and manipulate hypergraphs, using the Hypergraph Tree as memory to guide iterative analysis. Since agentic LLMs can automatically plan, execute, and explore complex tasks~\cite{zhang2025deepanalyze}, our future work includes building an agentic \explorativebi system to further lower the barrier to BI analysis.
\end{remark}
}
}

\section{Operator Implementation}
\label{sec:sample}

This section describes how to implement the operators in the computational layer of \sys introduced in the previous section. Specifically, since operators inherited from traditional BI, such as \drilldownop and \sliceop, can be naturally implemented following traditional BI approaches, we focus on the implementation of three new operators: \sourceop, \joinop, and \viewop.

\subsection{Basic Implementation of New Operators}
\label{sec:sample:basic}

\stitle{\sourceop Operator.} According to \refdef{source-operator}, the \sourceop operator performs subgraph matching on the property graph to find matching results, which form hyperedges, thereby obtaining a hypergraph. The implementation directly invokes the subgraph matching algorithm in the storage layer, as described in Section~\ref{sec:overview}.

\stitle{\joinop Operator.} The \joinop operator merges two hypergraphs into one, merging hyperedges based on the join condition in the process. An intuitive method is as follows: for two hypergraphs $\mathcal{G}_1 = (\mathcal{V}_{\mathcal{G}_1}, \mathcal{E}_{\mathcal{G}_1}, \kappa_{\mathcal{G}_1}, \mu_{\mathcal{G}_1})$ and $\mathcal{G}_2 = (\mathcal{V}_{\mathcal{G}_2}, \mathcal{E}_{\mathcal{G}_2}, \kappa_{\mathcal{G}_2}, \mu_{\mathcal{G}_2})$ participating in the join, for each hyperedge in $\mathcal{G}_1$, iterate through the hyperedges in $\mathcal{G}_2$ to determine whether these two hyperedges satisfy the join condition. When the join condition is satisfied, add the new hyperedge to the resulting hypergraph. This join implementation method is called \joinnosample. Note that join performance can be improved by building indexes on the vertex properties, edge properties, and hyperedge properties.

\stitle{\viewop Operator.} The \viewop operator transforms a hypergraph into a relational table, enabling the application of traditional BI methods for hypergraph analysis. Specifically, given a hypergraph, the \viewop operator iterates through its hyperedges, and for each hyperedge, projects it into a row in the relational table according to \refdef{view-operator}, thereby obtaining the final relational table.

\subsection{Performance Challenges}
\label{sec:sample:challenges}

The above implementations of \sourceop and \joinop face significant performance challenges in supporting interactive exploratory analysis. In \explorativebi, analysts typically need to repeatedly invoke these operators---for example, continuously calling \sourceop to obtain new hypergraphs, or calling \joinop multiple times to merge hypergraphs for iterative analysis. Therefore, the execution efficiency of these operators is critical to ensure smooth analytical workflows. Unfortunately, the straightforward implementations can suffer from critical performance issues.

\stitle{\sourceop Operator Challenges.} For the \sourceop operator, since subgraph matching is an NP-complete problem~\cite{qiao2017subgraph}, implementing an exact algorithm to find all matching results is very time-consuming. In our experiments, finding 4-path with four vertices in a graph with only 9,892 vertices takes about 4 minutes. Additionally, subgraph matching may produce a large number of results. If hyperedges corresponding to each result are all constructed and stored in the hypergraph, this will also lead to unacceptable memory overhead.

\stitle{\joinop Operator Challenges.} The \joinop operator faces similar problems. As mentioned above, when hyperedges in $\mathcal{G}_1$ and $\mathcal{G}_2$ are sorted according to the join condition, the time complexity of the \joinnosample algorithm is $O(|\mathcal{E}_{\mathcal{G}_1}| + |\mathcal{E}_{\mathcal{G}_2}|)$. However, as \explorativebi analysis progresses, the number of hyperedges in hypergraphs may increase dramatically, leading to poor performance of \joinnosample. Specifically, on the one hand, the \sourceop operator performing subgraph matching may result in a large number of hyperedges in the hypergraph; on the other hand, multiple joins may also cause the number of edges in the hypergraph to grow (especially when the join condition $\sigma = \emptyset$, i.e., performing a Cartesian product). Moreover, during the \explorativebi analysis process, sorting or building indexes for hyperedges in the hypergraph is itself time-consuming. When hyperedges in the hypergraph are unsorted, the time complexity of the \joinnosample algorithm rises to $O(|\mathcal{E}_{\mathcal{G}_1}| \cdot |\mathcal{E}_{\mathcal{G}_2}|)$, making it difficult to meet practical computation requirements.

\subsection{Sampling-Based Solutions}
\label{sec:sample:solutions}

To address the performance challenges faced by the \sourceop and \joinop operators, we adopt sampling-based approaches that compute results on a curated subset rather than exhaustively materializing all matches.
\revise{The key motivation for sampling-based solutions stems from the nature of BI analysis: analysts are typically interested in aggregated values (e.g., COUNT, SUM) rather than individual data records~\cite{aberger2018levelheaded}. Once these aggregates can be accurately approximated from samples using \emph{unbiased} estimators, computing the complete result set becomes unnecessary. We first introduce the sampling methods for bother \sourceop and \joinop operators.} 

\stitle{Sampling for \sourceop Operator.} For the \sourceop operator, we adapt the FaSTest algorithm \cite{FaSTest} to perform sampling during the subgraph matching process, resulting in a new implementation called \sourcesample. This sampling algorithm provides an unbiased consistent estimator for the number of subgraph matching results.

To adapt FaSTest to our hypergraph data model definition, we made several modifications. First, we modified the data structures to support subgraph matching on property graphs. Second, we adjusted the requirements for matching results to find the induced subgraphs isomorphisms. Third, unlike the original algorithm which only performs counting, we record the found matching results to construct hyperedges in the hypergraph.

\stitle{Sampling for \joinop Operator.} For the \joinop operator, we implement the \textbf{Strategy Group-Sample} algorithm proposed by Chaudhuri et al.~\cite{chaudhuri1999random}. Specifically, for a join operation $\mathcal{G}_1 \join_\sigma \mathcal{G}_2$, suppose the desired number of sampled join results is $r$. For each hyperedge $\epsilon$ in $\mathcal{G}_1$, we define its weight $w(\epsilon)$ as the number of hyperedges in $\mathcal{G}_2$ that can be joined with $\epsilon$ according to the join condition $\sigma$.

During sampling, we first sample $r$ hyperedges from $\mathcal{G}_1$ according to the weights $w(\epsilon)$ (denoted as $\epsilon_1, \ldots, \epsilon_r$). These hyperedges, together with the property graph from $\mathcal{G}_1$, form a new hypergraph $\mathcal{G}'_1$. Then, we perform $\mathcal{G}'_1 \join_{\sigma} \mathcal{G}_2$, sampling one result for each hyperedge in $\mathcal{G}'_1$. After this sampling, the obtained results are equivalent to performing a uniform sampling of size $r$ on the results of $\mathcal{G}_1 \join_{\sigma} \mathcal{G}_2$. We call this sampling-based implementation \joinsample.

Based on the above sampling-based \sourceop and \joinop operators, we can significantly improve the execution performance of these operators and meet real-time analysis requirements. Although the above \sourcesample and \joinsample algorithms are adapted from existing methods, a key challenge in applying them to the \sys system is ensuring end-to-end accurate estimation after integrating these sampling algorithms into the system.


\subsection{Unbiased Estimation Guarantee}
\label{sec:sample:unbiased}

In this subsection, we will prove that after applying these sampling algorithms, when computing aggregate functions like \countfunc and \sumfunc, we can obtain unbiased estimates of the true values, thereby demonstrating the feasibility of our approach.
Meanwhile, for other aggregate functions such as \distinctcountfunc, \quantilefunc, \maxfunc, and \minfunc, to the best of our knowledge, there are no established unbiased estimation methods. Therefore, we adopt existing approximate estimation techniques and try to ensure an error bound for the estimation: for \distinctcountfunc, we use the GEE estimator~\cite{charikar2000towards}; for \quantilefunc, we employ the t-digest method~\cite{dunning2019computing}.
Moreover, for \maxfunc and \minfunc, we use the maximum and minimum values from the sampled data as estimates.
Please note that these approximate estimation methods such as GEE estimator and t-digest method require uniformly sampled data as input for a better estimation.

To demonstrate that our sampling-based operators enable unbiased estimation for \countfunc and \sumfunc, we first establish the connection between uniform sampling and unbiased estimation through the following lemma.

\begin{lemma}
\label{lemma:uniform-sampling-unbiased}
Given a relational table $\hat{R}$ that is uniformly sampled from an original table $R$ with sampling rate $\rho$, an unbiased estimator for \countfunc is:
$$\sum_{r \in \hat{R}} \frac{1}{\rho}.$$
Similarly, for computing \sumfunc on attribute $X$, an unbiased estimator is:
$$\sum_{r \in \hat{R}} \frac{X_r}{\rho},$$
where $X_r$ denotes the value of attribute $X$ in row $r$.
\end{lemma}

\begin{proof}
This follows directly from the Horvitz-Thompson estimator~\cite{horvitz1952generalization}.
\end{proof}

This lemma implies that as long as the hypergraph used in the \viewop operator contains uniformly sampled hyperedges, the projected relational table is equivalent to a uniformly sampled table. Consequently, according to Lemma~\ref{lemma:uniform-sampling-unbiased}, we can compute unbiased estimates for \countfunc and \sumfunc. Therefore, we need to prove that the hyperedges in sampled hypergraphs satisfy uniform sampling. Specifically, the sampling-based operators include \sourcesample and \joinsample.

\stitle{Uniform Sampling Guarantee for \sourcesample.} For the \sourcesample operator, we establish the following lemma.

\begin{lemma}
\label{lemma:source-sample-uniform}
Given a property graph $G$ and query graph $q$, let $\mathcal{G}\left<G\right> = \sourcesample(G, q)$. Then, the hyperedges in $\mathcal{G}\left<G\right>$ are uniformly sampled from all subgraph matching results.
\end{lemma}

\begin{proof}
The FaSTest algorithm~\cite{FaSTest} implements two sampling algorithms for subgraph matching: tree sampling and graph sampling. The authors proved that tree sampling produces uniform samples of subgraph matching results, while graph sampling produces stratified samples. To ensure uniform sampling of subgraph matching results, we modified the FaSTest algorithm to use only the tree sampling algorithm. This guarantees that the results of \sourcesample are uniformly sampled.
\end{proof}

\stitle{Uniform Sampling Guarantee for \joinsample.} \revisesigmod{Regarding} \\ the \joinsample operator, we first prove the following lemma.

\begin{lemma}
\label{lemma:join-sample-uniform}
Given hypergraphs $\mathcal{G}_1 = (\mathcal{V}_{\mathcal{G}_1}, \mathcal{E}_{\mathcal{G}_1}, \kappa_{\mathcal{G}_1}, \mu_{\mathcal{G}_1})$ and $\mathcal{G}_2 = (\mathcal{V}_{\mathcal{G}_2}, \mathcal{E}_{\mathcal{G}_2}, \kappa_{\mathcal{G}_2}, \mu_{\mathcal{G}_2})$, and sample size $r$, let $\mathcal{G}_r = \mathcal{G}_1 \text{ } \joinsample_{\sigma} \mathcal{G}_2$ and $\mathcal{G}^t_r = \mathcal{G}_1 \text{ } \joinnosample_{\sigma} \mathcal{G}_2$. Then, each hyperedge in $\mathcal{G}_r$ is sampled from $\mathcal{G}^t_r$ with uniform probability $\frac{r}{W_{\mathcal{G}_1}}$, where $W_{\mathcal{G}_1} = \sum_{\epsilon' \in \mathcal{E}_{\mathcal{G}_1}} m_{\mathcal{G}_2, \sigma}(\epsilon')$ and $m_{\mathcal{G}_2, \sigma}(\epsilon')$ denotes the number of hyperedges in $\mathcal{G}_2$ that can join with $\epsilon'$ according to $\sigma$.
\end{lemma}

\begin{proof}
Suppose $\epsilon_r$ is a hyperedge in $\mathcal{G}_r$ and it is obtained by joining $\epsilon_1 \in \mathcal{E}_{\mathcal{G}_1}$ and $\epsilon_2 \in \mathcal{E}_{\mathcal{G}_2}$. Then, we analyze the probability of $\epsilon_r$ being sampled.

Firstly, according to the \joinsample algorithm, we need to sample $r$ hyperedges from $\mathcal{E}_{\mathcal{G}_1}$.
The probability that $\epsilon_1$ is sampled is:
$$P(\epsilon_1 | \mathcal{G}_1, \mathcal{G}_2) = r \cdot \frac{m_{\mathcal{G}_2, \sigma}(\epsilon_1)}{W_{\mathcal{G}_1}}.$$

After that, we need to select one hyperedge from $\mathcal{G}_2$ that can join with $\epsilon_1$ (denote it as $\epsilon_2$). The probability of selecting $\epsilon_2$ is:
$$P(\epsilon_2|\epsilon_1) = \frac{1}{m_{\mathcal{G}_2, \sigma}(\epsilon_1)}.$$

Therefore, the probability that $\epsilon_r$ is sampled is:
\begin{equation*}
    \begin{split}
        P(\epsilon_r|\epsilon_1, \epsilon_2) & = P(\epsilon_1|\mathcal{G}_1, \mathcal{G}_2) \cdot P(\epsilon_2|\epsilon_1) \\  
        & = r \cdot \frac{m_{\mathcal{G}_2, \sigma}(\epsilon_1)}{W_{\mathcal{G}_1}} \cdot \frac{1}{m_{\mathcal{G}_2, \sigma}(\epsilon_1)} \\
        & = \frac{r}{W_{\mathcal{G}_1}}.
    \end{split}
\end{equation*}

Since this probability is independent of the specific choice of $\epsilon_1$ and $\epsilon_2$, each hyperedge in $\mathcal{G}^t_r$ is sampled with equal probability $\frac{r}{W_{\mathcal{G}_1}}$, demonstrating that \joinsample provides uniform sampling.
\end{proof}

Next, we consider the scenario where the input hypergraphs $\mathcal{G}_1$ and $\mathcal{G}_2$ are also obtained through uniform sampling. In the following theorem, we use $\hat{\mathcal{G}}_1$ and $\hat{\mathcal{G}}_2$ to denote random variables representing the sampled hypergraphs, and $\overline{\mathcal{G}}_1$ and $\overline{\mathcal{G}}_2$ to denote a specific realization (observation) of the sampling. Specifically, we present the following theorem.

\begin{lemma}
\label{lemma:cascaded-join-sample}
Given hypergraphs $\mathcal{G}_1 = (\mathcal{V}_{\mathcal{G}_1}, \mathcal{E}_{\mathcal{G}_1}, \kappa_{\mathcal{G}_1}, \mu_{\mathcal{G}_1})$ and $\mathcal{G}_2 = (\mathcal{V}_{\mathcal{G}_2}, \mathcal{E}_{\mathcal{G}_2}, \kappa_{\mathcal{G}_2}, \mu_{\mathcal{G}_2})$, suppose we uniformly sample from $\mathcal{G}_1$ and $\mathcal{G}_2$ with sampling rates $\rho_1$ and $\rho_2$ to obtain a specific realization $\overline{\mathcal{G}}_1 = (\overline{\mathcal{V}}_{\overline{\mathcal{G}}_1}, \overline{\mathcal{E}}_{\overline{\mathcal{G}}_1})$ and $\overline{\mathcal{G}}_2$. Let $\overline{\mathcal{G}}_r = \overline{\mathcal{G}}_1 \text{ } \joinsample_{\sigma} \overline{\mathcal{G}}_2$ with sample size $r$, and $\mathcal{G}^t_r = \mathcal{G}_1 \text{ } \joinnosample_{\sigma} \mathcal{G}_2$. Then, $\overline{\mathcal{G}}_r$ is approximately a uniform sample of $\mathcal{G}^t_r$ with sampling rate $\frac{r \cdot \rho_1 \cdot \rho_2}{W_{\overline{\mathcal{G}}_1}}$, where $W_{\overline{\mathcal{G}}_1} = \sum_{\epsilon' \in \overline{\mathcal{E}}_{\overline{\mathcal{G}}_1}} m_{\overline{\mathcal{G}}_2, \sigma}(\epsilon')$.
\end{lemma}

\begin{proof}
Suppose $\epsilon_r$ is a hyperedge in $\mathcal{G}^t_r$, obtained by joining $\epsilon_1 \in \mathcal{E}_{\mathcal{G}_1}$ and $\epsilon_2 \in \mathcal{E}_{\mathcal{G}_2}$. We analyze the probability that $\epsilon_r$ is sampled in the final result.

Let the sampled hypergraphs be $\hat{\mathcal{G}}_1$ and $\hat{\mathcal{G}}_2$, according to Lemma~\ref{lemma:join-sample-uniform}, the probability that $\epsilon_r$ is sampled is:
$$\mathbf{1}_{\mathcal{E}_1 \in \hat{\mathcal{G}}_1} \cdot \mathbf{1}_{\mathcal{E}_2 \in \hat{\mathcal{G}}_2} \cdot \frac{r}{W_{\hat{\mathcal{G}}_1}},$$
where $\mathbf{1}_{\epsilon_1 \in \hat{\mathcal{G}}_1}$ is an indicator function satisfying:
$$\mathbf{1}_{\epsilon_1 \in \hat{\mathcal{G}}_1} = \begin{cases}
    1 & \text{if } \epsilon_1 \in \hat{\mathcal{G}}_1 \\
    0 & \text{if } \epsilon_1 \notin \hat{\mathcal{G}}_1
\end{cases}.$$

Since $\hat{\mathcal{G}}_1$ and $\hat{\mathcal{G}}_2$ are sampled with rates $\rho_1$ and $\rho_2$, respectively, we need to compute the expectation over different instances of $\hat{\mathcal{G}}_1$ and $\hat{\mathcal{G}}_2$. We can decompose this expectation by conditioning on whether $\epsilon_1$ and $\epsilon_2$ are sampled:
\begin{equation*}
    \begin{split}
        P( & \epsilon_r| \mathcal{G}_1, \mathcal{G}_2) = \mathbb{E}_{\hat{\mathcal{G}}_1, \hat{\mathcal{G}}_2}\left[\mathbf{1}_{\epsilon_1 \in \hat{\mathcal{G}}_1} \cdot \mathbf{1}_{\epsilon_2 \in \hat{\mathcal{G}}_2} \cdot \frac{r}{W_{\hat{\mathcal{G}}_1}}\right] \\
        & = P(\epsilon_1 \in \hat{\mathcal{G}}_1, \epsilon_2 \in \hat{\mathcal{G}}_2) \cdot \mathbb{E}_{\hat{\mathcal{G}}_1, \hat{\mathcal{G}}_2}\left[\frac{r}{W_{\hat{\mathcal{G}}_1}} \mid \epsilon_1 \in \hat{\mathcal{G}}_1, \epsilon_2 \in \hat{\mathcal{G}}_2\right] \\
        & \quad + P(\epsilon_1 \notin \hat{\mathcal{G}}_1 \text{ or } \epsilon_2 \notin \hat{\mathcal{G}}_2) \cdot 0 \\
        & = r \cdot \rho_1 \cdot \rho_2 \cdot \mathbb{E}_{\hat{\mathcal{G}}_1, \hat{\mathcal{G}}_2}\left[\frac{1}{W_{\hat{\mathcal{G}}_1}} \mid \epsilon_1 \in \hat{\mathcal{G}}_1, \epsilon_2 \in \hat{\mathcal{G}}_2\right].
    \end{split}
\end{equation*}

$\overline{\mathcal{G}}_1$ and $\overline{\mathcal{G}}_2$ are sampled from $\mathcal{G}_1$ and $\mathcal{G}_2$, respectively, and are specific realizations of the sampling.
Assuming that $\epsilon_1 \in \overline{\mathcal{G}}_1$ and $\epsilon_2 \in \overline{\mathcal{G}}_2$, we define the estimator:
$$\widehat{P}(\epsilon_r | \mathcal{G}_1, \mathcal{G}_2) = r \cdot \rho_1 \cdot \rho_2 \cdot \frac{1}{W_{\overline{\mathcal{G}}_1}}.$$

To show this is an unbiased estimator, we compute its expectation:
\begin{equation*}
    \begin{split}
        \mathbb{E}_{\hat{\mathcal{G}}_1, \hat{\mathcal{G}}_2} & \left[\widehat{P}(\epsilon_r | \mathcal{G}_1, \mathcal{G}_2) \mid \epsilon_1 \in \hat{\mathcal{G}}_1, \epsilon_2 \in \hat{\mathcal{G}}_2\right] \\
        & = \mathbb{E}_{\hat{\mathcal{G}}_1, \hat{\mathcal{G}}_2}\left[r \cdot \rho_1 \cdot \rho_2 \cdot \frac{1}{W_{\overline{\mathcal{G}}_1}} \mid \epsilon_1 \in \hat{\mathcal{G}}_1, \epsilon_2 \in \hat{\mathcal{G}}_2\right] \\
        & = r \cdot \rho_1 \cdot \rho_2 \cdot \mathbb{E}_{\hat{\mathcal{G}}_1, \hat{\mathcal{G}}_2}\left[\frac{1}{W_{\overline{\mathcal{G}}_1}} \mid \epsilon_1 \in \hat{\mathcal{G}}_1, \epsilon_2 \in \hat{\mathcal{G}}_2\right] \\
        & = r \cdot \rho_1 \cdot \rho_2 \cdot \mathbb{E}_{\hat{\mathcal{G}}_1, \hat{\mathcal{G}}_2}\left[\frac{1}{W_{\hat{\mathcal{G}}_1}} \mid \epsilon_1 \in \hat{\mathcal{G}}_1, \epsilon_2 \in \hat{\mathcal{G}}_2\right] \\
        & = P(\epsilon_r | \mathcal{G}_1, \mathcal{G}_2).
    \end{split}
\end{equation*}

Thus, $\overline{\mathcal{G}}_r$ is approximately a uniform sample of $\mathcal{G}^t_r$ with sampling rate $\frac{r \cdot \rho_1 \cdot \rho_2}{W_{\overline{\mathcal{G}}_1}}$.
\end{proof}

\stitle{End-to-End Uniform Sampling Guarantee.} Based on \\ \reflem{source-sample-uniform} and \reflem{cascaded-join-sample}, we now establish that the entire workflow produces uniformly sampled results.

\begin{theorem}
\label{thm:end-to-end-uniform}
The hyperedges in a hypergraph obtained through any combination of \sourcesample and \joinsample operators are uniformly sampled from the true results.
\end{theorem}

\begin{proof}
We prove this by induction on the workflow structure. If the workflow consists of only \sourcesample, the result is uniformly sampled according to Lemma~\ref{lemma:source-sample-uniform}.

Now consider a workflow that combines \sourcesample and \joinsample operators, with the final step being a \joinsample operation. The two input hypergraphs to this \joinsample are themselves results from either \sourcesample or \joinsample operations. By the induction hypothesis and \reflem{source-sample-uniform} and \reflem{cascaded-join-sample}, these input hypergraphs contain uniformly sampled hyperedges. According to \reflem{cascaded-join-sample}, the final output is also uniformly sampled.
\end{proof}

\stitle{Unbiased Estimation for Aggregate Functions.} Finally, we establish the main result for unbiased estimation.

\begin{theorem}
\label{thm:unbiased-estimation}
Let $\mathcal{G}$ be a hypergraph obtained through any combination of \sourcesample and \joinsample operators. Then, unbiased estimators can be computed for \countfunc and \sumfunc based on $\mathcal{G}$.
\end{theorem}

\begin{proof}
This follows directly from Theorem~\ref{thm:end-to-end-uniform} and Lemma~\ref{lemma:uniform-sampling-unbiased}.
\end{proof}

Moreover, according to Theorem~\ref{thm:end-to-end-uniform}, since the results obtained through \sourcesample and \joinsample are uniformly sampled, the approximate estimation algorithms such as the GEE estimator and t-digest can also provide more accurate estimates.

\begin{table}[htbp]
    \centering
    \caption{Dataset statistics. $|\text{Vertex}|$ and $|\text{Edge}|$ denote the number of vertices and edges, respectively.}
    \label{tab:datasets}
    \renewcommand{\arraystretch}{1.2}
    \begin{tabular}{lrr}
    \toprule
    \textbf{Dataset} & \textbf{$|\text{Vertex}|$} & \textbf{$|\text{Edge}|$} \\
    \midrule
    SF0.1  & 327,588    & 1,477,965  \\ SF3    & 9,281,922  & 52,695,735  \\
    SF1    & 3,181,724  & 17,256,038 \\ SF10   & 29,987,835 & 176,623,445 \\
    
    \bottomrule
    \end{tabular}
\end{table}

\section{Experimental Evaluation}
\label{sec:experiment}

\begin{figure*}[htbp]
    \centering
    \includegraphics[width=\textwidth]{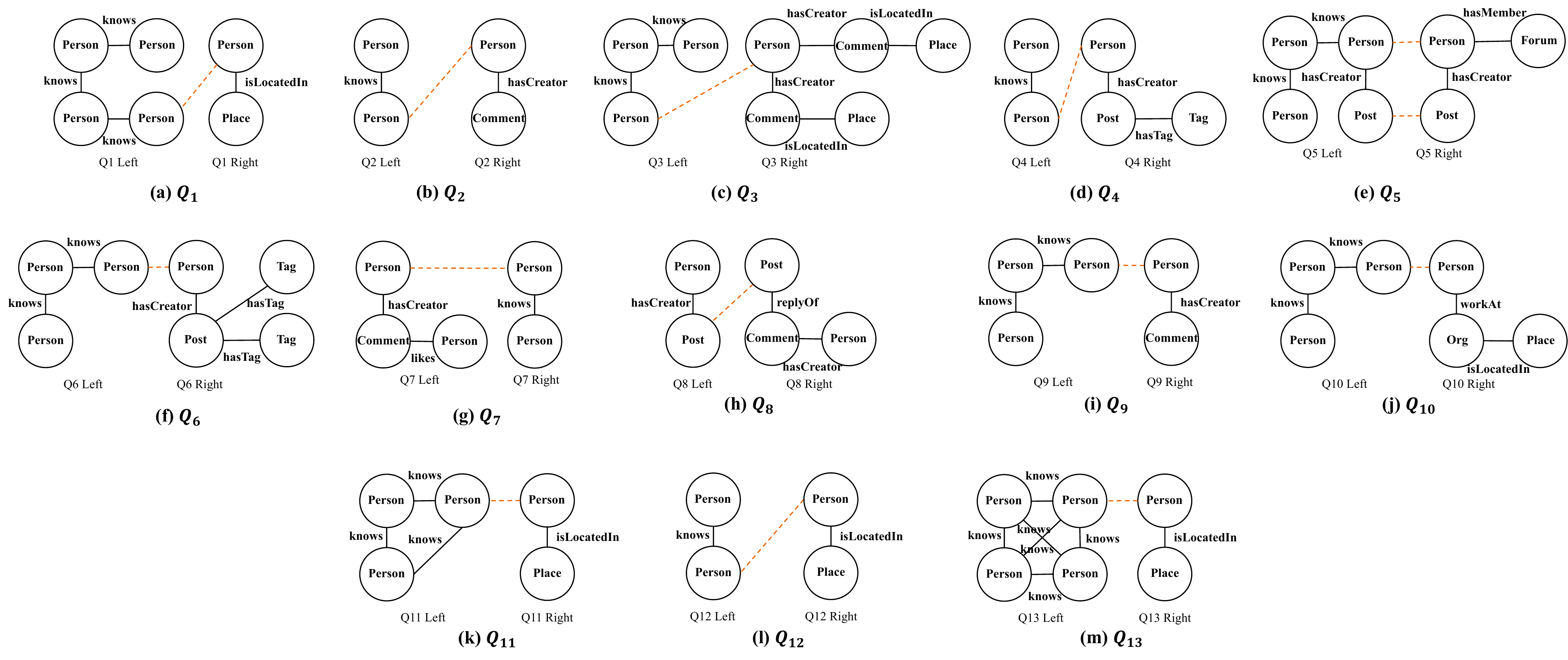}
    \caption{\explorativebi queries used in the experiments}
    \label{fig:exp:queries}
\end{figure*}

In this section, we first describe the experimental settings (\refsec{exp:settings}), including datasets, queries, and compared methods.
We then evaluate the performance of \sys through extensive experiments \refsec{exp:sample_size}--\refsec{exp:multi_join}.

\begin{figure*}[htbp]
    \centering
    \begin{subfigure}[b]{0.8\textwidth}
        \includegraphics[width=\textwidth]{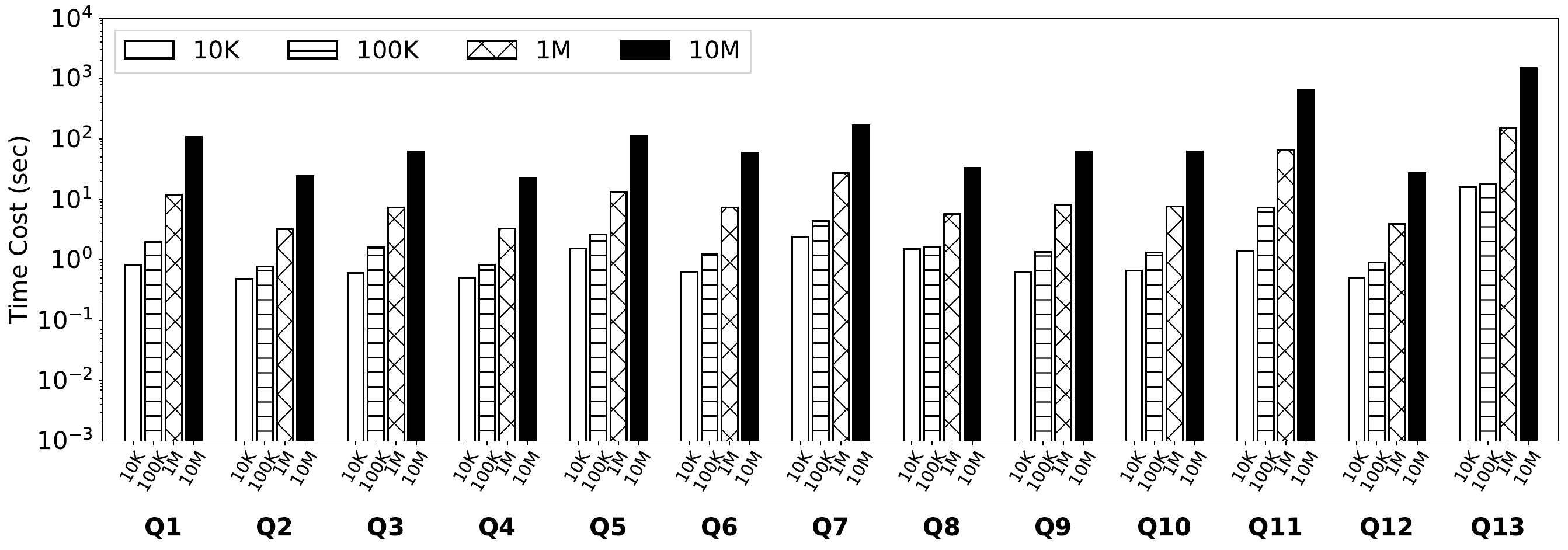}
        \caption{\sourcesample operator performance}
        \label{fig:exp:cube1}
    \end{subfigure}

    \begin{subfigure}[b]{0.8\textwidth}
        \includegraphics[width=\textwidth]{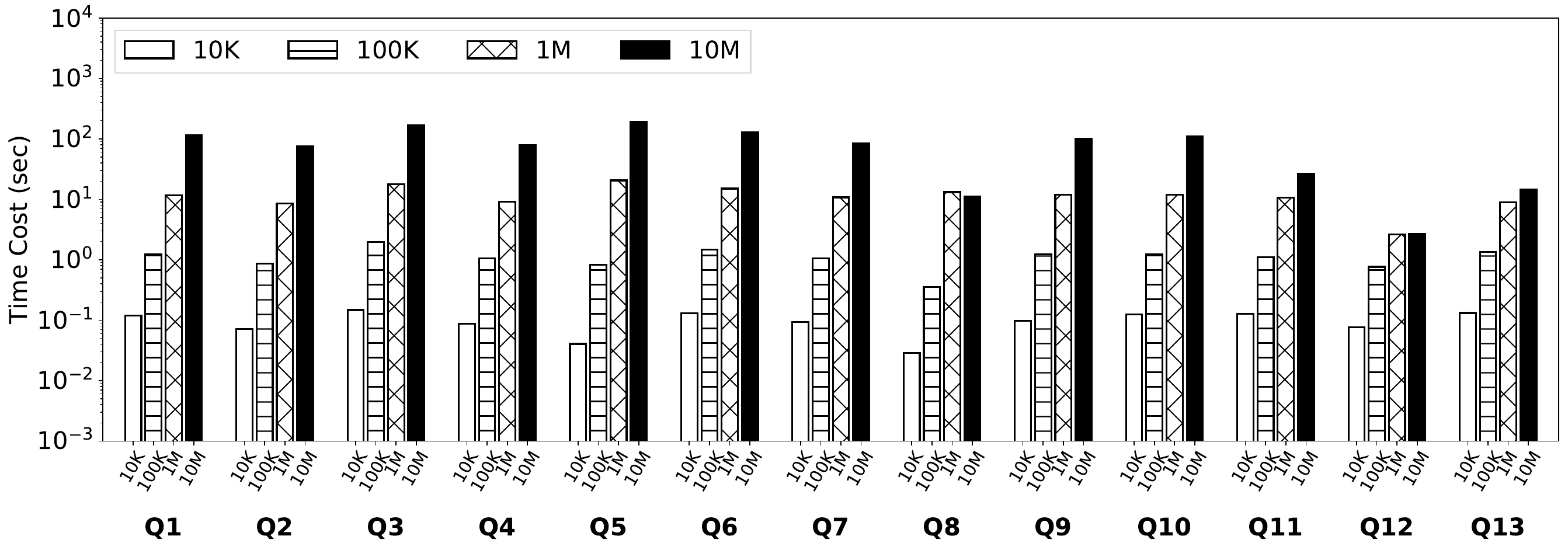}
        \caption{\joinsample operator performance}
        \label{fig:exp:join}
    \end{subfigure}

    \begin{subfigure}[b]{0.8\textwidth}
        \includegraphics[width=\textwidth]{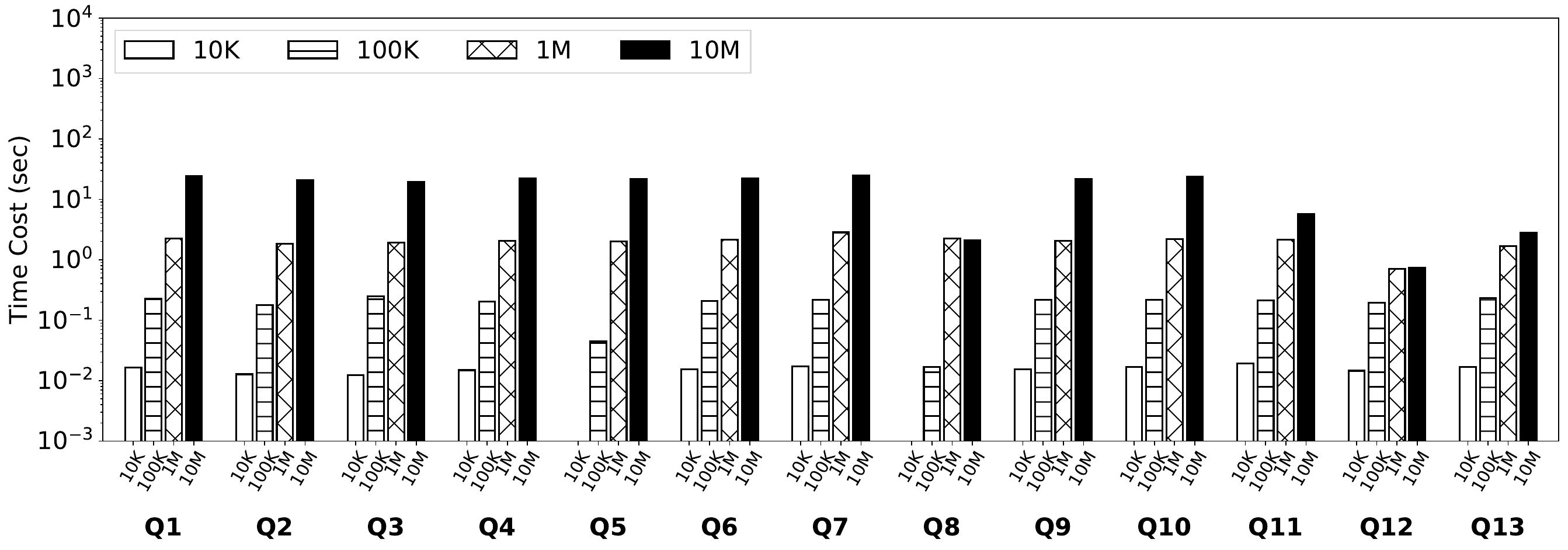}
        \caption{Analytical performance}
        \label{fig:exp:view}
    \end{subfigure}
    \caption{\centering Performance evaluation with varying sample size}
    \label{fig:exp:scale_performance}
\end{figure*}

\subsection{Experimental settings}
\label{sec:exp:settings}

\noindent
{\em \underline{Dataset.}} 
We use the LDBC Social Network Benchmarking (SNB)~\cite{ldbc_benchmarks} dataset for our experiments, as LDBC is widely adopted for graph algorithms benchmarking~\cite{lai2019distributed,lou2024towards} and provides large-scale property graphs suitable for our evaluation. 
\revisesigmod{The dataset consists of varying vertices (e.g., \texttt{Person}, \texttt{Post}, \texttt{Comment}, \texttt{Place}) and edges (e.g., \texttt{knows}, \texttt{hasCreator}, \texttt{isLocatedIn}). Among them, the \texttt{knows} table (representing friendship relations among persons) and the \texttt{hasCreator} table (representing authorship of posts/comments) are two of the largest relations.}
We evaluate our approach on four scale factors: SF0.1, SF1, SF3, and SF10.
The details of the datasets are shown in \reftable{datasets}, and the largest dataset SF10 contains approximately 30 million vertices and 177 million edges.
\revisesigmod{We do not extend our experiments to larger-scale datasets because non-sampling approaches (e.g., \sourcenosample and \joinnosample) already fail in many cases to complete queries at SF10. Consequently, obtaining the exact ground truth required to evaluate the sampling error can be infeasible at larger datasets.}

\noindent
{\em \underline{Exploratory BI queries.}} 
Since there are no publicly available exploratory BI query benchmarks, we design 13 queries (Q1–Q13) based on IC queries~\cite{ldbc_benchmarks} on the LDBC datasets to evaluate our approach. 
The queries are shown in \reffig{exp:queries}. 
Each query consists of two query graphs connected through join conditions, which are depicted as yellow dashed lines. 
For instance, in Q1 (\reffig{exp:queries}(a)), the \texttt{Person} vertex in the left query graph (Q1 Left) is joined with the \texttt{Person} vertex in the right query graph (Q1 Right) on \texttt{Person.ID}, as indicated by the yellow dashed line.
Note that for each query, we perform \sourceop on both the left and right query graphs separately. Since the two source operations are similar in nature, we report the results for the left query graph by default in our experiments.
Furthermore, we evaluate commonly used aggregate functions including \distinctcountfunc, \countfunc, and \maxfunc for each query. 
\revisesigmod{We limit most experiments to joins between two query graphs due to the computational cost of exact approaches. In \refsec{exp:multi_join}, we further evaluate the performance of multiple joins.}


\noindent
{\em \underline{Compared methods.}} 
We compare the following methods in our experiments:
(1) \textbf{\sys}: Our full system with sampling-based operators (\sourcesample and \joinsample) for efficient approximate query processing.
(2) \textbf{\sysnosample}: A variant of \sys that replaces \sourcesample and \joinsample with their exact counterparts \sourcenosample and \joinnosample, serving as an ablation baseline.
(3) \textbf{MySQL \cite{mysql8}}: A widely-used traditional relational database system.
(4) \textbf{Neo4j \cite{neo4j}}: A popular graph database system.
Moreover, we use a sampling-based method as a baseline, i.e., \textbf{VerdictDB \cite{verdictdb}}, which is an approximate relational query processing middleware that provides fast approximate answers with error bounds. We use VerdictDB as the representative of pre-sampling AQP systems (including BlinkDB~\cite{blinkdb}), since both rely on pre-computed sampling views and share similar methodological foundations.
For MySQL and Neo4j, we adopt their default configurations and perform system warm-up. For VerdictDB, we configure MySQL as its backend database.

\revisesigmod{For the comparative evaluation on MySQL, Neo4j, and VerdictDB, we execute standalone queries (manually composed corresponding to those in \reffig{exp:queries}) directly on these baseline systems, rather than deploying the full \explorativebi framework on top of them. While a full integration is theoretically feasible (\refsec{overview:storage}), we opt for standalone execution to highlight the limitations of these engines. First, MySQL and Neo4j lack native sampling support. As demonstrated in our experiments (\refsec{exp:e2e}), executing exact subgraph matching without sampling incurs prohibitive latency, failing to meet the interactive requirements of exploratory analysis. Second, although VerdictDB supports sampling, its capabilities are limited for our scenarios. We observe poor performance on complex queries (e.g. Q13) and limitations in handling the ``non-existence'' semantics essential for induced subgraph queries.}

\noindent
{\em \underline{Error rate calculation.}} 
To measure the accuracy of our sampling-based approaches, we calculate the error rate as: $\text{Error Rate} = \\ \frac{|\text{True Value} - \text{Estimated Value}|}{\text{True Value}} \times 100\%$, where the true value is obtained from non-sampling methods (e.g., \sysnosample and Neo4j) when they successfully complete, and the estimated value is obtained from sampling-based methods (\sys and VerdictDB).

\noindent
{\em \underline{Implementation and environment.}} 
\sys and \sysnosample are implemented in C++. For the baseline systems, we use MySQL 8.0, Neo4j 2025.11.2, and pyverdict 0.1.3.2, a Python client for VerdictDB.
All experiments are executed on a Linux server with an Intel Platinum 8269CY processor and 512 GB of main memory. 
{\em We set a timeout of 2 hours for each query execution. The abbreviations \texttt{OOT} and \texttt{OOM} denote out-of-time and out-of-memory errors, respectively.}

\subsection {Effect of sample size}
\label{sec:exp:sample_size}

To balance efficiency and accuracy, we evaluate the impact of sample size on execution time and error rates using four sizes (10K, 100K, 1M, 10M) on SF1. \reffig{exp:scale_performance} shows the running times of \sourceop, \joinop, and the complete analysis process (including \viewop and aggregation). Tables~\ref{tab:exp:source_est}--\ref{tab:exp:source_count} report error rates for different operators and aggregation functions. 
Note that \sysnosample cannot complete queries $Q_1$, $Q_3$, $Q_5$, $Q_6$, and $Q_9$ due to OOM errors; however, Neo4j successfully completes some of these queries (specifically $Q_6$ and $Q_9$), so we combine the true values from both methods where available.
Additionally, \reftable{exp:source_est} contains true values for $Q_3$ while Tables~\ref{tab:exp:source_distinct_count}--\ref{tab:exp:source_count} do not. This is because
\revisesigmod{This is because \sourcenosample completed successfully on Q3, allowing us to obtain the true value for source estimation, while \joinnosample ran out of memory (OOM), preventing us from obtaining the true value for the aggregate functions.}

As shown in \reffig{exp:scale_performance}, execution time increases almost linearly with sample size, demonstrating the predictable scalability of the saampling approaches.
For example, in query $Q_5$, increasing the sample size from 1M to 10M causes the execution time to increase by approximately 9--10$\times$ for all 
operators (\sourcesample: 13.19s to 397.79s, \joinsample: 20.55s to 191.39s).

For accuracy, \reftable{exp:source_est} shows that \sourcesample~maintains error rates below 2\% across all queries and sample sizes (e.g., $Q_{10}$ achieves 0.00\% error at 1M samples, $Q_{13}$ achieves 0.15\% error at 1M samples). 
For aggregation functions, \reftable{exp:source_count} shows that \countfunc maintains low error rates (mostly below 1\%), and \reftable{exp:source_max} shows that \maxfunc achieves near-zero error rates (0.00\%--0.01\%) across all queries. 
For \distinctcountfunc (\reftable{exp:source_distinct_count}), error rates are generally higher but improve significantly with larger samples (e.g., $Q_4$ error decreases from 124.51\% at 100K to 2.94\% at 1M, and $Q_7$ error decreases from 106.44\% at 100K to 2.37\% at 1M).
It is worth noting that the estimation error for \distinctcountfunc is considerably higher than that for \countfunc. This is because, to the best of our knowledge, there is no unbiased estimator for \distinctcountfunc in the literature. Our use of the biased GEE estimator~\cite{charikar2000towards} thus incurs relatively larger errors.

Based on this analysis, we select 1M as the default sample size, which provides an optimal balance between execution time and accuracy—increasing to 10M yields only marginal accuracy improvements while execution time increases by nearly 10$\times$. 

\begin{table*}[htbp]
    \centering
    \small
    \setlength{\tabcolsep}{3pt}
    \caption{\sourcesample estimation under different sample sizes}
    \label{tab:exp:source_est}
    \scalebox{0.7}{
    \begin{tabular}{l l *{8}{l}}
    \toprule
    \multirow{2}{*}{\textbf{Query}} & 
    \multirow{2}{*}{\textbf{True value}} & 
    \multicolumn{8}{c}{\textbf{Sample size}} \\
    \cmidrule(lr){3-10}
    & & \multicolumn{2}{c}{\textbf{10K}} & \multicolumn{2}{c}{\textbf{100K}} & \multicolumn{2}{c}{\textbf{1M}} & \multicolumn{2}{c}{\textbf{10M}} \\
    \cmidrule(lr){3-4} \cmidrule(lr){5-6} \cmidrule(lr){7-8} \cmidrule(lr){9-10}
    & & \textbf{estimate value} & \textbf{error rate} & \textbf{estimate value} & \textbf{error rate} & \textbf{estimate value} & \textbf{error rate} & \textbf{estimate value} & \textbf{error rate} \\
    \midrule
    Q1 & -- & $3.56{\scriptstyle \pm 0.01} \times 10^{9}$ & -- & $3.57{\scriptstyle \pm 0.00} \times 10^{9}$ & -- & $3.57{\scriptstyle \pm 0.00} \times 10^{9}$ & -- & $3.57{\scriptstyle \pm 0.00} \times 10^{9}$ & -- \\
    Q2 & $3.61 \times 10^{5}$ & $3.61{\scriptstyle \pm 0.00} \times 10^{5}$ & 0.00\% & $3.61{\scriptstyle \pm 0.00} \times 10^{5}$ & 0.00\% & $3.61{\scriptstyle \pm 0.00} \times 10^{5}$ & 0.00\% & $3.61{\scriptstyle \pm 0.00} \times 10^{5}$ & 0.00\% \\
    Q3 & $4.85 \times 10^{7}$ & $4.85{\scriptstyle \pm 0.01} \times 10^{7}$ & 0.08\% & $4.85{\scriptstyle \pm 0.00} \times 10^{7}$ & 0.04\% & $4.85{\scriptstyle \pm 0.00} \times 10^{7}$ & 0.00\% & $4.85{\scriptstyle \pm 0.00} \times 10^{7}$ & 0.00\% \\
    Q4 & $3.61 \times 10^{5}$ & $3.61{\scriptstyle \pm 0.00} \times 10^{5}$ & 0.00\% & $3.61{\scriptstyle \pm 0.00} \times 10^{5}$ & 0.00\% & $3.61{\scriptstyle \pm 0.00} \times 10^{5}$ & 0.00\% & $3.61{\scriptstyle \pm 0.00} \times 10^{5}$ & 0.00\% \\
    Q5 & -- & $6.23{\scriptstyle \pm 0.02} \times 10^{9}$ & -- & $6.23{\scriptstyle \pm 0.01} \times 10^{9}$ & -- & $6.23{\scriptstyle \pm 0.00} \times 10^{9}$ & -- & $6.22{\scriptstyle \pm 0.00} \times 10^{9}$ & -- \\
    Q6 & $4.85 \times 10^{7}$ & $4.85{\scriptstyle \pm 0.01} \times 10^{7}$ & 0.11\% & $4.85{\scriptstyle \pm 0.00} \times 10^{7}$ & 0.02\% & $4.85{\scriptstyle \pm 0.00} \times 10^{7}$ & 0.00\% & $4.85{\scriptstyle \pm 0.00} \times 10^{7}$ & 0.01\% \\
    Q7 & $3.32 \times 10^{5}$ & $3.33{\scriptstyle \pm 0.03} \times 10^{5}$ & 0.23\% & $3.33{\scriptstyle \pm 0.02} \times 10^{5}$ & 0.06\% & $3.32{\scriptstyle \pm 0.00} \times 10^{5}$ & 0.04\% & $3.32{\scriptstyle \pm 0.00} \times 10^{5}$ & 0.02\% \\
    Q8 & $1.00 \times 10^{6}$ & $1.00{\scriptstyle \pm 0.00} \times 10^{6}$ & 0.01\% & $1.00{\scriptstyle \pm 0.00} \times 10^{6}$ & 0.00\% & $1.00{\scriptstyle \pm 0.00} \times 10^{6}$ & 0.00\% & $1.00{\scriptstyle \pm 0.00} \times 10^{6}$ & 0.00\% \\
    Q9 & $4.85 \times 10^{7}$ & $4.85{\scriptstyle \pm 0.01} \times 10^{7}$ & 0.04\% & $4.84{\scriptstyle \pm 0.00} \times 10^{7}$ & 0.08\% & $4.85{\scriptstyle \pm 0.00} \times 10^{7}$ & 0.01\% & $4.85{\scriptstyle \pm 0.00} \times 10^{7}$ & 0.00\% \\
    Q10 & $4.85 \times 10^{7}$ & $4.85{\scriptstyle \pm 0.01} \times 10^{7}$ & 0.04\% & $4.85{\scriptstyle \pm 0.00} \times 10^{7}$ & 0.02\% & $4.85{\scriptstyle \pm 0.00} \times 10^{7}$ & 0.00\% & $4.85{\scriptstyle \pm 0.00} \times 10^{7}$ & 0.00\% \\
    Q11 & $2.33 \times 10^{6}$ & $2.32{\scriptstyle \pm 0.03} \times 10^{6}$ & 0.28\% & $2.33{\scriptstyle \pm 0.01} \times 10^{6}$ & 0.02\% & $2.32{\scriptstyle \pm 0.00} \times 10^{6}$ & 0.05\% & $2.33{\scriptstyle \pm 0.00} \times 10^{6}$ & 0.01\% \\
    Q12 & $3.61 \times 10^{5}$ & $3.61{\scriptstyle \pm 0.00} \times 10^{5}$ & 0.00\% & $3.61{\scriptstyle \pm 0.00} \times 10^{5}$ & 0.00\% & $3.61{\scriptstyle \pm 0.00} \times 10^{5}$ & 0.00\% & $3.61{\scriptstyle \pm 0.00} \times 10^{5}$ & 0.00\% \\
    Q13 & $4.15 \times 10^{6}$ & $4.09{\scriptstyle \pm 0.44} \times 10^{6}$ & 1.51\% & $4.15{\scriptstyle \pm 0.09} \times 10^{6}$ & 0.11\% & $4.15{\scriptstyle \pm 0.03} \times 10^{6}$ & 0.15\% & $4.15{\scriptstyle \pm 0.01} \times 10^{6}$ & 0.04\% \\
    \bottomrule
    \end{tabular}
    }
\end{table*}

\begin{table*}[htbp]
    \centering
    \small
    \setlength{\tabcolsep}{3pt}
    \caption{\distinctcountfunc estimation results under different sample sizes}
    \label{tab:exp:source_distinct_count}
    \scalebox{0.7}{
    \begin{tabular}{l l *{8}{l}}
    \toprule
    \multirow{2}{*}{\textbf{Query}} & 
    \multirow{2}{*}{\textbf{True value}} & 
    \multicolumn{8}{c}{\textbf{Sample size}} \\
    \cmidrule(lr){3-10}
    & & \multicolumn{2}{c}{\textbf{10K}} & \multicolumn{2}{c}{\textbf{100K}} & \multicolumn{2}{c}{\textbf{1M}} & \multicolumn{2}{c}{\textbf{10M}} \\
    \cmidrule(lr){3-4} \cmidrule(lr){5-6} \cmidrule(lr){7-8} \cmidrule(lr){9-10}
    & & \textbf{estimate value} & \textbf{error rate} & \textbf{estimate value} & \textbf{error rate} & \textbf{estimate value} & \textbf{error rate} & \textbf{estimate value} & \textbf{error rate} \\
    \midrule
    Q1 & -- & $424547.00 \, (\pm 10797.10)$ & -- & $46357.60 \, (\pm 2566.43)$ & -- & $6279.80 \, (\pm 143.01)$ & -- & $5782.40 \, (\pm 11.61)$ & -- \\
    Q2 & 8902 & $101006.20 \, (\pm 3943.92)$ & 1034.65\% & $57268.60 \, (\pm 1681.72)$ & 543.32\% & $19226.40 \, (\pm 260.67)$ & 115.98\% & $9566.00 \, (\pm 85.09)$ & 7.46\% \\
    Q3 & -- & $6164258.20 \, (\pm 246404.01)$ & -- & $3986300.20 \, (\pm 232410.09)$ & -- & $1370220.60 \, (\pm 26064.86)$ & -- & $282901.60 \, (\pm 11122.68)$ & -- \\
    Q4 & 9144 & $65062.80 \, (\pm 3192.51)$ & 611.54\% & $20529.40 \, (\pm 480.24)$ & 124.51\% & $9413.20 \, (\pm 65.81)$ & 2.94\% & $8932.40 \, (\pm 19.17)$ & 2.31\% \\
    Q5 & -- & $5718803.80 \, (\pm 460339.78)$ & -- & $6790257.80 \, (\pm 604904.39)$ & -- & $1316865.00 \, (\pm 10578.52)$ & -- & $320473.20 \, (\pm 8470.13)$ & -- \\
    Q6 & 6074 & $860720.00 \, (\pm 21818.38)$ & 14070.56\% & $382346.40 \, (\pm 21508.26)$ & 6194.80\% & $50611.40 \, (\pm 2286.56)$ & 733.25\% & $8124.40 \, (\pm 315.27)$ & 33.76\% \\
    Q7 & 9186 & $54137.20 \, (\pm 1096.13)$ & 489.34\% & $18963.40 \, (\pm 476.83)$ & 106.44\% & $9403.60 \, (\pm 30.57)$ & 2.37\% & $9043.40 \, (\pm 7.13)$ & 1.55\% \\
    Q8 & 8231 & $3094.40 \, (\pm 495.94)$ & 62.41\% & $8703.20 \, (\pm 213.31)$ & 5.74\% & $7638.80 \, (\pm 23.06)$ & 7.19\% & $7649.00 \, (\pm 10.27)$ & 7.07\% \\
    Q9 & 8902 & $991158.80 \, (\pm 28850.72)$ & 11034.11\% & $518665.40 \, (\pm 4741.14)$ & 5726.39\% & $86820.40 \, (\pm 2503.80)$ & 875.29\% & $15159.40 \, (\pm 160.04)$ & 70.29\% \\
    Q10 & 7413 & $60139.00 \, (\pm 2952.96)$ & 711.26\% & $15076.60 \, (\pm 176.67)$ & 103.38\% & $6532.20 \, (\pm 39.78)$ & 11.88\% & $6427.40 \, (\pm 19.35)$ & 13.30\% \\
    Q11 & 8322 & $7948.80 \, (\pm 291.87)$ & 4.48\% & $5173.00 \, (\pm 107.42)$ & 37.84\% & $5045.00 \, (\pm 24.26)$ & 39.38\% & $5164.80 \, (\pm 27.87)$ & 37.94\% \\
    Q12 & 9163 & $5198.00 \, (\pm 154.90)$ & 43.27\% & $5258.00 \, (\pm 37.30)$ & 42.62\% & $5574.20 \, (\pm 29.93)$ & 39.17\% & $5569.80 \, (\pm 20.99)$ & 39.21\% \\
    Q13 & 6595 & $8344.60 \, (\pm 706.82)$ & 26.53\% & $3648.40 \, (\pm 82.60)$ & 44.68\% & $3767.40 \, (\pm 49.78)$ & 42.87\% & $4060.25 \, (\pm 32.87)$ & 38.43\% \\
    \bottomrule
    \end{tabular}
    }
    \end{table*}
    
    \begin{table*}[htbp]
    \centering
    \small
    \setlength{\tabcolsep}{3pt}
    \caption{\maxfunc estimation results under different sample sizes}
    \label{tab:exp:source_max}
    \scalebox{0.7}{
    \begin{tabular}{l l *{8}{l}}
    \toprule
    \multirow{2}{*}{\textbf{Query}} & 
    \multirow{2}{*}{\textbf{True value}} & 
    \multicolumn{8}{c}{\textbf{Sample size}} \\
    \cmidrule(lr){3-10}
    & & \multicolumn{2}{c}{\textbf{10K}} & \multicolumn{2}{c}{\textbf{100K}} & \multicolumn{2}{c}{\textbf{1M}} & \multicolumn{2}{c}{\textbf{10M}} \\
    \cmidrule(lr){3-4} \cmidrule(lr){5-6} \cmidrule(lr){7-8} \cmidrule(lr){9-10}
    & & \textbf{estimate value} & \textbf{error rate} & \textbf{estimate value} & \textbf{error rate} & \textbf{estimate value} & \textbf{error rate} & \textbf{estimate value} & \textbf{error rate} \\
    \midrule
    Q1 & -- & $2160391.27 \, (\pm 15.17)$ & -- & $2160399.50 \, (\pm 11.46)$ & -- & $2160396.46 \, (\pm 5.36)$ & -- & $2160394.84 \, (\pm 8.61)$ & -- \\
    Q2 & 2160506 & $2160327.00 \, (\pm 0.00)$ & 0.01\% & $2160327.00 \, (\pm 0.01)$ & 0.01\% & $2160327.01 \, (\pm 0.01)$ & 0.01\% & $2160327.01 \, (\pm 0.01)$ & 0.01\% \\
    Q3 & -- & $2160327.00 \, (\pm 0.00)$ & -- & $2160327.00 \, (\pm 0.00)$ & -- & $2160327.00 \, (\pm 0.00)$ & -- & $2160327.00 \, (\pm 0.00)$ & -- \\
    Q4 & 2160507 & $2160372.62 \, (\pm 37.71)$ & 0.01\% & $2160388.44 \, (\pm 3.17)$ & 0.01\% & $2160396.73 \, (\pm 2.88)$ & 0.01\% & $2160394.93 \, (\pm 5.71)$ & 0.01\% \\
    Q5 & -- & $2160327.00 \, (\pm 0.00)$ & -- & $2160327.00 \, (\pm 0.00)$ & -- & $2160327.00 \, (\pm 0.00)$ & -- & $2160327.00 \, (\pm 0.00)$ & -- \\
    Q6 & 2160507 & $2160338.07 \, (\pm 7.34)$ & 0.01\% & $2160331.65 \, (\pm 0.73)$ & 0.01\% & $2160332.98 \, (\pm 0.90)$ & 0.01\% & $2160332.43 \, (\pm 0.36)$ & 0.01\% \\
    Q7 & 2160507 & $2160382.11 \, (\pm 13.30)$ & 0.01\% & $2160398.76 \, (\pm 5.01)$ & 0.01\% & $2160393.84 \, (\pm 3.83)$ & 0.01\% & $2160396.85 \, (\pm 4.36)$ & 0.01\% \\
    Q8 & 2160507 & $2160329.68 \, (\pm 104.02)$ & 0.01\% & $2160364.81 \, (\pm 24.15)$ & 0.01\% & $2160364.14 \, (\pm 4.69)$ & 0.01\% & $2160361.27 \, (\pm 4.30)$ & 0.01\% \\
    Q9 & 2160506 & $2160331.51 \, (\pm 4.50)$ & 0.01\% & $2160328.07 \, (\pm 0.71)$ & 0.01\% & $2160327.93 \, (\pm 0.65)$ & 0.01\% & $2160327.85 \, (\pm 0.44)$ & 0.01\% \\
    Q10 & 2160506 & $2160402.11 \, (\pm 12.27)$ & 0.00\% & $2160403.79 \, (\pm 4.24)$ & 0.00\% & $2160398.90 \, (\pm 5.99)$ & 0.00\% & $2160405.32 \, (\pm 1.77)$ & 0.00\% \\
    Q11 & 2160507 & $2160378.56 \, (\pm 23.31)$ & 0.01\% & $2160391.05 \, (\pm 10.79)$ & 0.01\% & $2160379.64 \, (\pm 22.42)$ & 0.01\% & $2160383.45 \, (\pm 11.72)$ & 0.01\% \\
    Q12 & 2160507 & $2160382.19 \, (\pm 22.34)$ & 0.01\% & $2160392.11 \, (\pm 5.79)$ & 0.01\% & $2160387.20 \, (\pm 5.47)$ & 0.01\% & $2160393.35 \, (\pm 7.08)$ & 0.01\% \\
    Q13 & 2160507 & $2160382.60 \, (\pm 7.19)$ & 0.01\% & $2160377.92 \, (\pm 12.81)$ & 0.01\% & $2160356.29 \, (\pm 10.87)$ & 0.01\% & $2160354.12 \, (\pm 22.46)$ & 0.01\% \\
    \bottomrule
    \end{tabular}
    }
    \end{table*}
    
    \begin{table*}[htbp]
    \centering
    \small
    \setlength{\tabcolsep}{3pt}
    \caption{\countfunc estimation results under different sample sizes}
    \label{tab:exp:source_count}
    \scalebox{0.7}{
    \begin{tabular}{l l *{8}{l}}
    \toprule
    \multirow{2}{*}{\textbf{Query}} & 
    \multirow{2}{*}{\textbf{True value}} & 
    \multicolumn{8}{c}{\textbf{Sample size}} \\
    \cmidrule(lr){3-10}
    & & \multicolumn{2}{c}{\textbf{10K}} & \multicolumn{2}{c}{\textbf{100K}} & \multicolumn{2}{c}{\textbf{1M}} & \multicolumn{2}{c}{\textbf{10M}} \\
    \cmidrule(lr){3-4} \cmidrule(lr){5-6} \cmidrule(lr){7-8} \cmidrule(lr){9-10}
    & & \textbf{estimate value} & \textbf{error rate} & \textbf{estimate value} & \textbf{error rate} & \textbf{estimate value} & \textbf{error rate} & \textbf{estimate value} & \textbf{error rate} \\
    \midrule
    Q1 & -- & $3.53{\scriptstyle \pm 0.05} \times 10^{9}$ & -- & $3.55{\scriptstyle \pm 0.02} \times 10^{9}$ & -- & $3.59{\scriptstyle \pm 0.02} \times 10^{9}$ & -- & $3.55{\scriptstyle \pm 0.03} \times 10^{9}$ & -- \\
    Q2 & $2.69 \times 10^{8}$ & $2.68{\scriptstyle \pm 0.06} \times 10^{8}$ & 0.64\% & $2.69{\scriptstyle \pm 0.02} \times 10^{8}$ & 0.20\% & $2.69{\scriptstyle \pm 0.01} \times 10^{8}$ & 0.01\% & $2.69{\scriptstyle \pm 0.00} \times 10^{8}$ & 0.01\% \\
    Q3 & -- & $3.16{\scriptstyle \pm 0.10} \times 10^{12}$ & -- & $3.07{\scriptstyle \pm 0.05} \times 10^{12}$ & -- & $3.08{\scriptstyle \pm 0.01} \times 10^{12}$ & -- & $3.10{\scriptstyle \pm 0.00} \times 10^{12}$ & -- \\
    Q4 & $7.49 \times 10^{7}$ & $7.53{\scriptstyle \pm 0.08} \times 10^{7}$ & 0.50\% & $7.52{\scriptstyle \pm 0.10} \times 10^{7}$ & 0.32\% & $7.50{\scriptstyle \pm 0.01} \times 10^{7}$ & 0.17\% & $7.48{\scriptstyle \pm 0.03} \times 10^{7}$ & 0.08\% \\
    Q5 & -- & $4.63{\scriptstyle \pm 0.15} \times 10^{12}$ & -- & $4.65{\scriptstyle \pm 0.05} \times 10^{12}$ & -- & $4.64{\scriptstyle \pm 0.00} \times 10^{12}$ & -- & $4.64{\scriptstyle \pm 0.00} \times 10^{12}$ & -- \\
    Q6 & $2.73 \times 10^{10}$ & $2.72{\scriptstyle \pm 0.09} \times 10^{10}$ & 0.22\% & $2.71{\scriptstyle \pm 0.01} \times 10^{10}$ & 0.80\% & $2.73{\scriptstyle \pm 0.01} \times 10^{10}$ & 0.17\% & $2.73{\scriptstyle \pm 0.00} \times 10^{10}$ & 0.04\% \\
    Q7 & $2.92 \times 10^{7}$ & $2.91{\scriptstyle \pm 0.10} \times 10^{7}$ & 0.35\% & $2.94{\scriptstyle \pm 0.02} \times 10^{7}$ & 0.42\% & $2.93{\scriptstyle \pm 0.01} \times 10^{7}$ & 0.08\% & $2.92{\scriptstyle \pm 0.01} \times 10^{7}$ & 0.12\% \\
    Q8 & $9.77 \times 10^{5}$ & $9.97{\scriptstyle \pm 1.27} \times 10^{5}$ & 2.07\% & $9.77{\scriptstyle \pm 0.15} \times 10^{5}$ & 0.00\% & $9.75{\scriptstyle \pm 0.04} \times 10^{5}$ & 0.14\% & $9.78{\scriptstyle \pm 0.01} \times 10^{5}$ & 0.09\% \\
    Q9 & $2.21 \times 10^{10}$ & $2.22{\scriptstyle \pm 0.03} \times 10^{10}$ & 0.19\% & $2.20{\scriptstyle \pm 0.02} \times 10^{10}$ & 0.72\% & $2.21{\scriptstyle \pm 0.00} \times 10^{10}$ & 0.03\% & $2.21{\scriptstyle \pm 0.00} \times 10^{10}$ & 0.00\% \\
    Q10 & $1.06 \times 10^{8}$ & $1.06{\scriptstyle \pm 0.02} \times 10^{8}$ & 0.28\% & $1.06{\scriptstyle \pm 0.00} \times 10^{8}$ & 0.12\% & $1.06{\scriptstyle \pm 0.01} \times 10^{8}$ & 0.42\% & $1.06{\scriptstyle \pm 0.01} \times 10^{8}$ & 0.24\% \\
    Q11 & $2.33 \times 10^{6}$ & $2.35{\scriptstyle \pm 0.07} \times 10^{6}$ & 1.18\% & $2.37{\scriptstyle \pm 0.07} \times 10^{6}$ & 2.11\% & $2.28{\scriptstyle \pm 0.07} \times 10^{6}$ & 2.04\% & $2.29{\scriptstyle \pm 0.03} \times 10^{6}$ & 1.40\% \\
    Q12 & $3.61 \times 10^{5}$ & $3.66{\scriptstyle \pm 0.05} \times 10^{5}$ & 1.30\% & $3.61{\scriptstyle \pm 0.06} \times 10^{5}$ & 0.07\% & $3.60{\scriptstyle \pm 0.06} \times 10^{5}$ & 0.34\% & $3.59{\scriptstyle \pm 0.02} \times 10^{5}$ & 0.66\% \\
    Q13 & $4.15 \times 10^{6}$ & $4.01{\scriptstyle \pm 0.47} \times 10^{6}$ & 3.36\% & $4.25{\scriptstyle \pm 0.10} \times 10^{6}$ & 2.43\% & $4.13{\scriptstyle \pm 0.13} \times 10^{6}$ & 0.64\% & $4.24{\scriptstyle \pm 0.14} \times 10^{6}$ & 2.06\% \\
    \bottomrule
    \end{tabular}
    }
    \end{table*}


\subsection{Comparison with Exact Methods}
\label{sec:exp:e2e}

In this section, we compare the end-to-end performance of \sys with exact methods (\sysnosample, MySQL, and Neo4j) on the SF1, SF3, and SF10 datasets.
As shown in \reffig{exp:scale_datasets}, \sys completes all 13 queries across all datasets within reasonable time, with execution times ranging from 10.12s ($Q_{12}$ on SF1) to 335.94s ($Q_3$ on SF10).
In contrast, exact methods increasingly fail as dataset size grows.
\sysnosample encounters \texttt{OOM} errors: 5/13 queries fail on SF1, 7/13 on SF3, and 9/13 on SF10. \revisesigmod{This is primarily due to the combinatorial explosion of intermediate results inherent in exact subgraph matching. For instance, $Q_6$ on SF1 produces over 27 billion results, making the full materialization of exact answers prohibitively memory-intensive.}
Neo4j, as a native graph database, performs relatively well but still encounters \texttt{OOT} errors on complex queries: 3/13 fail on SF1, 6/13 on SF3, and 6/13 on SF10.
MySQL shows worse scalability due to the complexity of multi-way joins: 6/13 queries fail on SF1, 10/13 on SF3, and 11/13 on SF10.
These results show that many queries are difficult to complete within reasonable time using exact algorithms, validating the necessity of sampling-based approaches.
\revisesigmod{On queries where both MySQL and Neo4j complete, Neo4j consistently outperforms MySQL by 3--260$\times$ (e.g., $Q_4$: 29.81s vs. 2911.05s; $Q_{10}$: 46.75s vs. 6009.28s; $Q_{11}$: 18.90s vs. 4935.93s on SF1). This performance gap underscores the inherent efficiency of the graph data model for structural queries, validating our choice to build \sys upon a graph-based foundation.}

\revisesigmod{In comparison with Neo4j, \sys delivers an average speedup of 16.21$\times$ and a maximum speedup of 146.25$\times$ (for $Q_9$ on SF1).
Compared to MySQL, \sys achieves higher speedups, i.e., an average speedup of 46.67$\times$ and a maximum speedup of 230.53$\times$ (for $Q_{10}$ on SF1).
While Neo4j exhibits lower latency on simple queries (e.g., $Q_{12}$ on SF1), this is an artifact of our fixed experimental configuration rather than a system limitation. Specifically, for queries with small ground-truth cardinalities (e.g., 361k for $Q_{12}$ on SF1), our default 1M sample size results in over-sampling, incurring processing overhead with no accuracy gain. By aligning the sample size with the query selectivity (e.g., 10K samples), \sys reduces latency to 1.02s (with only 1.30\% error). We will study query-adaptive sampling techniques in future work.}

\reftable{exp:aggregate_all} shows the estimation error rates of \sys across different datasets and aggregate functions.
Note that $Q_1$, $Q_3$, and $Q_5$ are not reported in this table because exact algorithms fail to complete on all three datasets, making it impossible to obtain ground truth values for error calculation.
For \countfunc and \maxfunc, \sys achieves consistently low error rates across all queries and datasets (mostly below 1\%).
For \distinctcountfunc, \sys shows higher error rates due to the use of a biased estimator (\refsec{exp:sample_size}).
The results demonstrate that \sys maintains high accuracy while successfully completing all 13 queries that exact methods cannot handle.

Overall, these results show that \sys significantly reduces execution time while maintaining comparable accuracy compared to the baselines, making it highly suitable for practical complex query workloads.

\begin{figure*}[htbp]
    \centering
    \begin{subfigure}[b]{\textwidth}
        \centering
        \includegraphics[width=.85\textwidth]{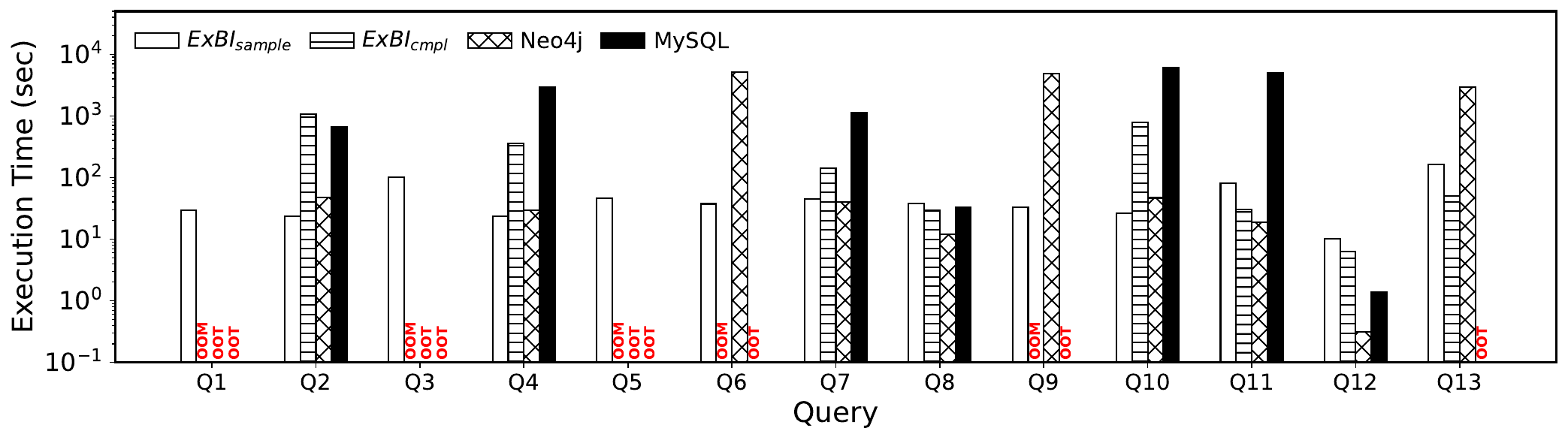}
        \caption{SF1}
        \label{fig:exp:sf1}
    \end{subfigure}
    \hfill
    \begin{subfigure}[b]{\textwidth}
        \centering
        \includegraphics[width=.85\textwidth]{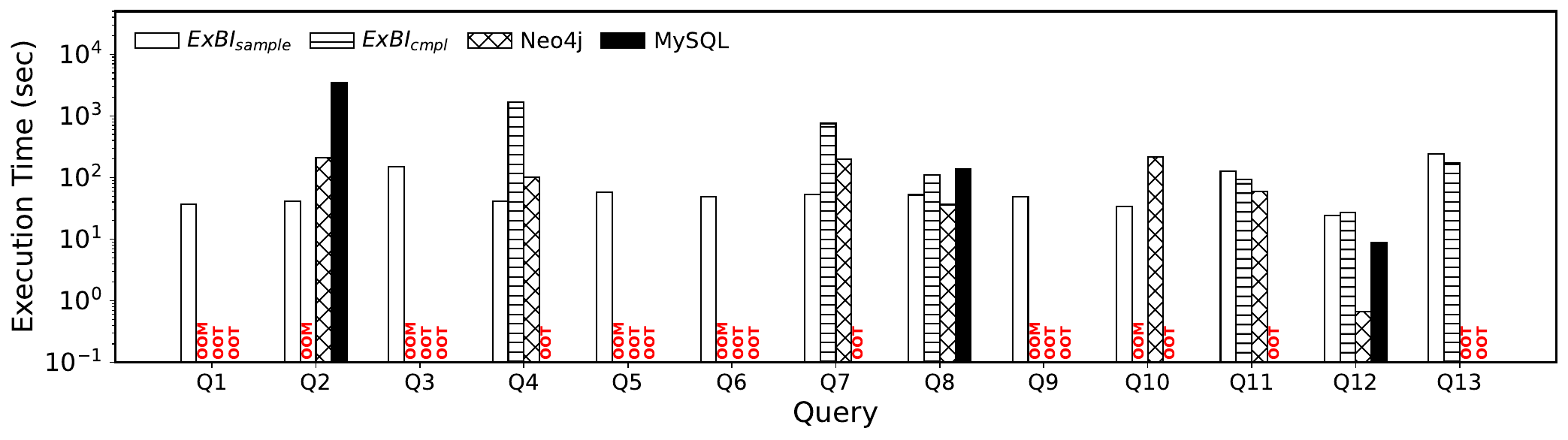}
        \caption{SF3}
        \label{fig:exp:sf3}
    \end{subfigure}
    \hfill
    \begin{subfigure}[b]{\textwidth}
        \centering
        \includegraphics[width=.85\textwidth]{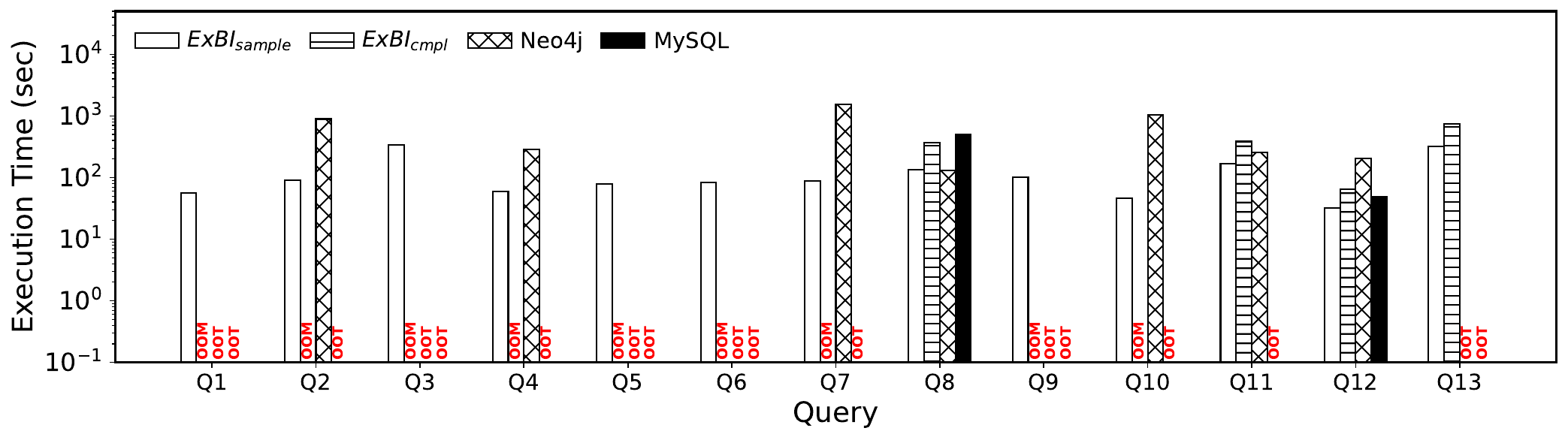}
        \caption{SF10}
        \label{fig:exp:sf10}
    \end{subfigure}
    \caption{\centering End-to-End performance evaluation on different datasets with sample size 1M. NA indicates that the query is not supported by the system.}
    \label{fig:exp:scale_datasets}
\end{figure*} 

\begin{table*}[htbp]
    \centering
    \small
    \caption{The estimation error rates on different datasets}
    \label{tab:exp:aggregate_all}
    \begin{tabular}{lccccccccc}
    \toprule
    \multirow{2}{*}{\textbf{Query}} & \multicolumn{3}{c}{\textbf{SF1}} & \multicolumn{3}{c}{\textbf{SF3}} & \multicolumn{3}{c}{\textbf{SF10}} \\
    \cmidrule(lr){2-4} \cmidrule(lr){5-7} \cmidrule(lr){8-10}
     & \distinctcountfunc & \maxfunc & \countfunc & \distinctcountfunc & \maxfunc & \countfunc & \distinctcountfunc & \maxfunc & \countfunc \\
    \midrule
    Q2 & 115.98\% & 0.01\% & 0.01\% & 316.49\% & 0.00\% & 0.09\% & 732.27\% & 0.00\% & 0.11\% \\
    Q4 & 2.94\% & 0.01\% & 0.17\% & 17.47\% & 0.00\% & 0.11\% & 106.40\% & 0.00\% & 0.15\% \\
    Q6 & 733.25\% & 0.01\% & 0.17\% & -- & -- & -- & -- & -- & -- \\
    Q7 & 2.37\% & 0.01\% & 0.08\% & 22.55\% & 0.00\% & 0.04\% & 142.95\% & 0.00\% & 0.09\% \\
    Q8 & 7.19\% & 0.01\% & 0.14\% & 14.91\% & 0.00\% & 0.16\% & 16.71\% & 0.00\% & 0.08\% \\
    Q9 & 875.29\% & 0.01\% & 0.03\% & -- & -- & -- & -- & -- & -- \\
    Q10 & 11.88\% & 0.00\% & 0.42\% & 0.28\% & 0.00\% & 0.00\% & 95.08\% & 0.00\% & 0.21\% \\
    Q11 & 39.38\% & 0.01\% & 2.04\% & 40.25\% & 0.00\% & 0.10\% & 37.59\% & 0.00\% & 0.09\% \\
    Q12 & 39.17\% & 0.01\% & 0.34\% & 38.86\% & 0.00\% & 0.86\% & 41.05\% & 0.00\% & 0.13\% \\
    Q13 & 42.87\% & 0.01\% & 0.64\% & 45.80\% & 0.00\% & 0.58\% & 43.54\% & 0.00\% & 0.21\% \\
    \bottomrule
    \end{tabular}
\end{table*}

\subsection{Comparison with VerdictDB}
\label{sec:exp:verdict}

In this section, we compare \sys with VerdictDB, a representative sampling-based AQP system.

\stitle{VerdictDB configurations.}
\revisesigmod{While \sys applies sampling dynamically during query execution, VerdictDB pre-computes materialized sample tables with fixed sampling rates.}
We evaluate three VerdictDB configurations:
\begin{itemize}
    \item \textbf{Ver$_{1\%}$}: Follows the recommended settings in~\cite{verdictdb}, pre-constructing sample tables with 1\% uniform samples and 1\% universe samples for large tables;
    \item \textbf{Ver$_{10\%}$}: Pre-constructs sample tables with 10\% sampling rate for more accurate results;
    \item \textbf{Ver$_{all}$}: Uses 100\% data. VerdictDB employs early stopping during computation, which terminates when the estimated error falls below a threshold, so Ver$_{all}$ is still not an exact algorithm.
\end{itemize}
On small datasets (SF1 and SF3), all tables do not qualify as ``large'' so VerdictDB skips pre-sampling and operates on the original tables. Consequently, all configurations reduce to Ver$_{all}$ and produce identical results. On SF10, VerdictDB's pre-sampling strategy activates on the large tables (e.g. \texttt{knows} table), causing different results across configurations.

\begin{table*}[htbp]
    \centering
    \small
    \caption{Comparison of \sys and VerdictDB on SF1 and SF3. cnt and dc are abbreviations for \countfunc and \distinctcountfunc, respectively. }
    \label{tab:exbi_verdict_comparison}
    \setlength{\tabcolsep}{2pt}
    \begin{tabular}{lcccccccccccccccc}
    \toprule
    \multirow{3}{*}{\textbf{Query}} & \multicolumn{4}{c}{\textbf{Time (s)}} & \multicolumn{6}{c}{\textbf{SF1 Error Rate (\%)}} & \multicolumn{6}{c}{\textbf{SF3 Error Rate (\%)}} \\
    \cmidrule(lr){2-5} \cmidrule(lr){6-11} \cmidrule(lr){12-17}
     & \multicolumn{2}{c}{SF1} & \multicolumn{2}{c}{SF3} & \multicolumn{3}{c}{\sys} & \multicolumn{3}{c}{Ver} & \multicolumn{3}{c}{\sys} & \multicolumn{3}{c}{Ver} \\
    \cmidrule(lr){2-3} \cmidrule(lr){4-5} \cmidrule(lr){6-8} \cmidrule(lr){9-11} \cmidrule(lr){12-14} \cmidrule(lr){15-17}
     & \sys & Ver & \sys & Ver & cnt & max & dc & cnt & max & dc & cnt & max & dc & cnt & max & dc \\
    \midrule
    Q11 & 80.17 & 63.93 & 126.86 & 239.13 & 2.04\% & 0.01\% & 39.38\% & 0.00\% & 0.00\% & 0.00\% & 0.10\% & 0.00\% & 40.25\% & 0.33\% & 0.00\% & 96.37\% \\
    Q12 & 10.12 & 1.00 & 24.10 & 3.95 & 0.34\% & 0.01\% & 39.17\% & 0.00\% & 0.00\% & 0.00\% & 0.86\% & 0.00\% & 38.86\% & 0.00\% & 0.00\% & 0.00\% \\
    Q13 & 163.19 & 6622.34 & 239.76 & OOT & 0.64\% & 0.01\% & 42.87\% & 0.00\% & 0.00\% & 0.00\% & 0.58\% & 0.00\% & 45.80\% & OOT & OOT & OOT \\
    \bottomrule
    \end{tabular}
\end{table*}

\begin{table}[htbp]
    \centering
    \small
    \caption{Comparison of \sys and VerdictDB configurations on SF10}
    \label{tab:exbi_verdict_sf10}
    \setlength{\tabcolsep}{1.2pt}
    \begin{tabular}{llcccc}
    \toprule
    \textbf{Metric} & \textbf{Query} & \textbf{\sys} & \textbf{Ver$_{1\%}$} & \textbf{Ver$_{10\%}$} & \textbf{Ver$_{all}$} \\
    \midrule
    \multirow{3}{*}{Time (s)} 
        & Q11 & 167.55 & 0.58 & 15.05 & 340.62 \\
        & Q12 & 32.28 & 0.33 & 1.61 & 6.84 \\
        & Q13 & 316.51 & 0.54 & OOT & OOT \\
    \midrule
    \multirow{3}{*}{\shortstack[l]{Error Rate of \\ \countfunc}} 
        & Q11 & 0.09\% & NA & 99.92\% & 0.21\% \\
        & Q12 & 0.13\% & 99.19\% & 90.78\% & 0.06\% \\
        & Q13 & 0.21\% & NA & OOT & OOT \\
    \midrule
    \multirow{3}{*}{\shortstack[l]{Error Rate of \\ \maxfunc }} 
        & Q11 & 0.00\% & NA & 0.00\% & 0.00\% \\
        & Q12 & 0.00\% & 0.00\% & 0.00\% & 0.00\% \\
        & Q13 & 0.00\% & NA & OOT & OOT \\
    \midrule
    \multirow{3}{*}{\shortstack[l]{Error Rate of \\ \distinctcountfunc}} 
        & Q11 & 37.59\% & 3.44e5\% & 8516.39\% & 661.04\% \\
        & Q12 & 41.05\% & 2.68\% & 0.66\% & 0.07\% \\
        & Q13 & 43.54\% & 1.30e10\% & OOT & OOT \\
    \bottomrule
    \end{tabular}
\end{table}

\revisesigmod{Unlike \sys, which computes all aggregates (\countfunc, \maxfunc, and \distinctcountfunc) in one query, VerdictDB necessitates different sample types for different aggregates: universe samples for \distinctcountfunc and uniform samples for \countfunc and \maxfunc. Consequently, VerdictDB cannot compute all aggregates in a single pass. In our experiments, we execute two separate queries for VerdictDB---one for \distinctcountfunc and another for \countfunc and \maxfunc---and report the maximum of the two execution times as VerdictDB's query time. If either query times out, the result is recorded as \texttt{OOT}}.
\revisesigmod{Note that we only report VerdictDB results for $Q_{11}$--$Q_{13}$ because VerdictDB cannot handle the ``non-existence'' semantics essential for induced subgraph queries required by $Q_1$--$Q_{10}$. VerdictDB can execute $Q_{11}$--$Q_{13}$ because their left and right query graphs are complete graphs, where ``non-existence'' checks are implicitly satisfied.}

\stitle{Results on SF1, SF3 (\reftable{exbi_verdict_comparison}).}
\revisesigmod{On the simple query $Q_{12}$, \sys is slower than VerdictDB because of the sampling overhead on small workloads.}
However, for complex queries like $Q_{13}$ (4-clique) and $Q_{11}$ (cyclic), \sys scales much better. VerdictDB times out on $Q_{13}$ in SF3 and becomes nearly $2\times$ slower than \sys on $Q_{11}$ in SF3.
Regarding accuracy, VerdictDB achieves near-perfect estimates on SF1. However, it exhibits extreme instability on SF3: the \distinctcountfunc error for $Q_{11}$ surges to 96.37\%, and we anticipate even larger errors for the timed-out $Q_{13}$. In contrast, \sys demonstrates robust performance: it maintains consistently low error rates for \countfunc and \maxfunc (mostly below 1\%) across all queries. Although \sys relies on a biased estimator for \distinctcountfunc (as discussed in \refsec{exp:sample_size}) which results in $\sim$40\% error, this error rate remains consistent across dataset scales, avoiding the unpredictable accuracy collapse observed in VerdictDB.

\stitle{Results on SF10 (\reftable{exbi_verdict_sf10}).}
\sys efficiently completes all three queries with stable error rates: \countfunc errors remain below 0.21\%, \maxfunc errors are consistently 0.00\%, and \distinctcountfunc errors \revisesigmod{lie between 37.59\% and 43.54\%} (due to the biased estimator).
In contrast, VerdictDB exhibits critical limitations.
First, When pre-sampling is not used, VerdictDB faces severe scalability bottlenecks: Ver$_{all}$ times out on the complex query $Q_{13}$ that involves a 4-clique.
Second, while pre-sampling can alleviate these scalability issues, it introduces structural fragility. Ver$_{1\%}$ returns ``None'' results (marked as \texttt{NA}) for \countfunc and \maxfunc on $Q_{11}$ and $Q_{13}$, because the 1\% uniform sampling rate breaks the connectivity of cyclic subgraphs. This hypothesis is confirmed by Ver$_{10\%}$: increasing the sampling rate restores enough connectivity to produce results for $Q_{11}$.
Third, even in cases where VerdictDB produces outputs, it still suffers from extreme statistical instability. For example, Ver$_{10\%}$ returns a result for $Q_{11}$, but with a 99.92\% \countfunc error. Similarly, the universe sampling strategy used for \distinctcountfunc yields significant errors: $3.44 \times 10^{5}\%$ for $Q_{11}$ and $1.30 \times 10^{10}\%$ for $Q_{13}$ under Ver$_{1\%}$, rendering the results useless in practice.

Comparatively, \sys offers a much more robust solution for complex queries required for \explorativebi, balancing performance and accuracy where baselines fail.

\begin{figure}[bp]
    \centering
    \begin{subfigure}[b]{0.45\textwidth}
        \centering
        \includegraphics[width=\textwidth]{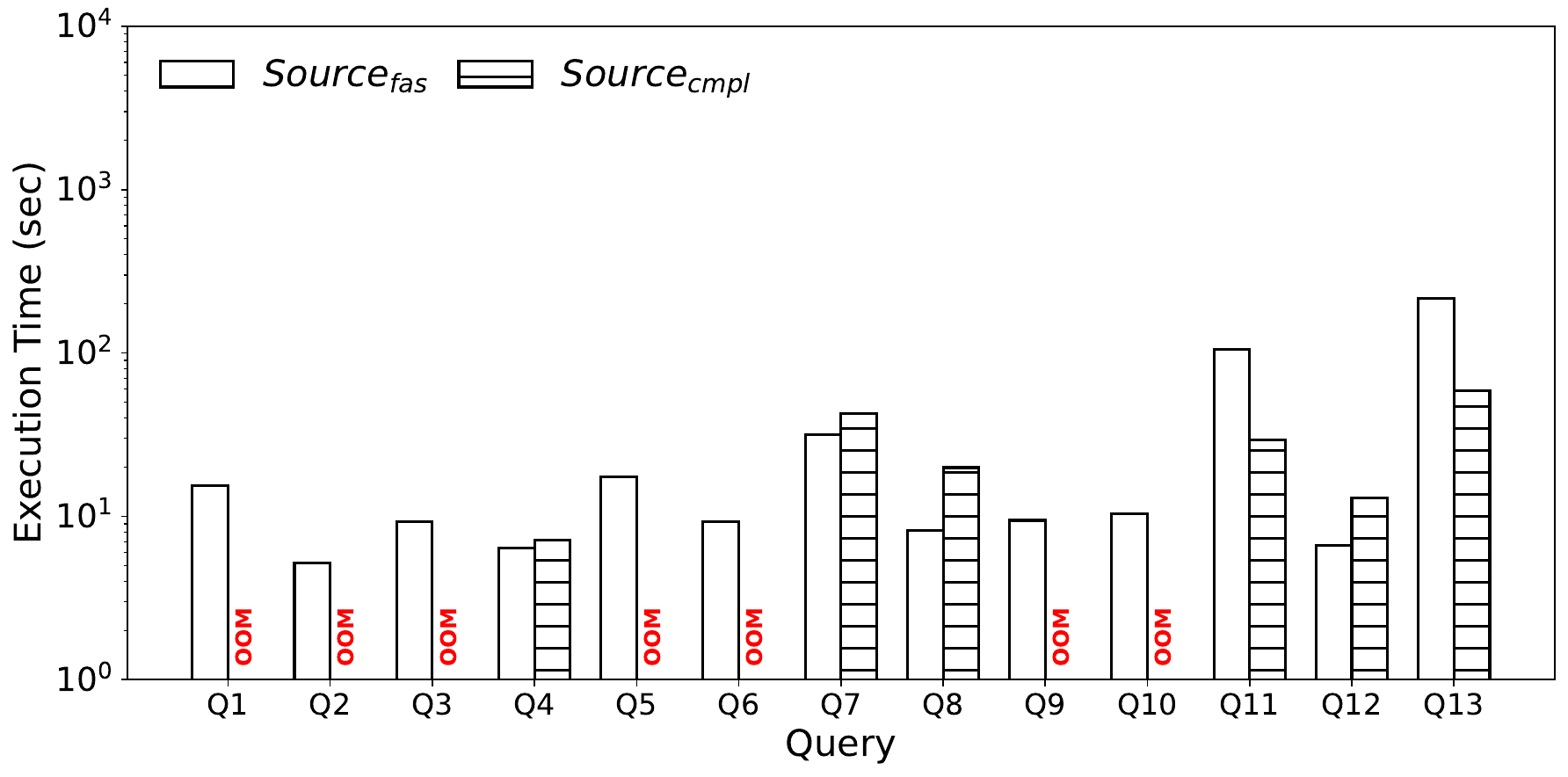}
        \caption{\sourcesample vs. \sourcenosample}
        \label{fig:exp:sourcesample_sf3}
    \end{subfigure}
    \hfill
    \begin{subfigure}[b]{0.45\textwidth}
        \centering
        \includegraphics[width=\textwidth]{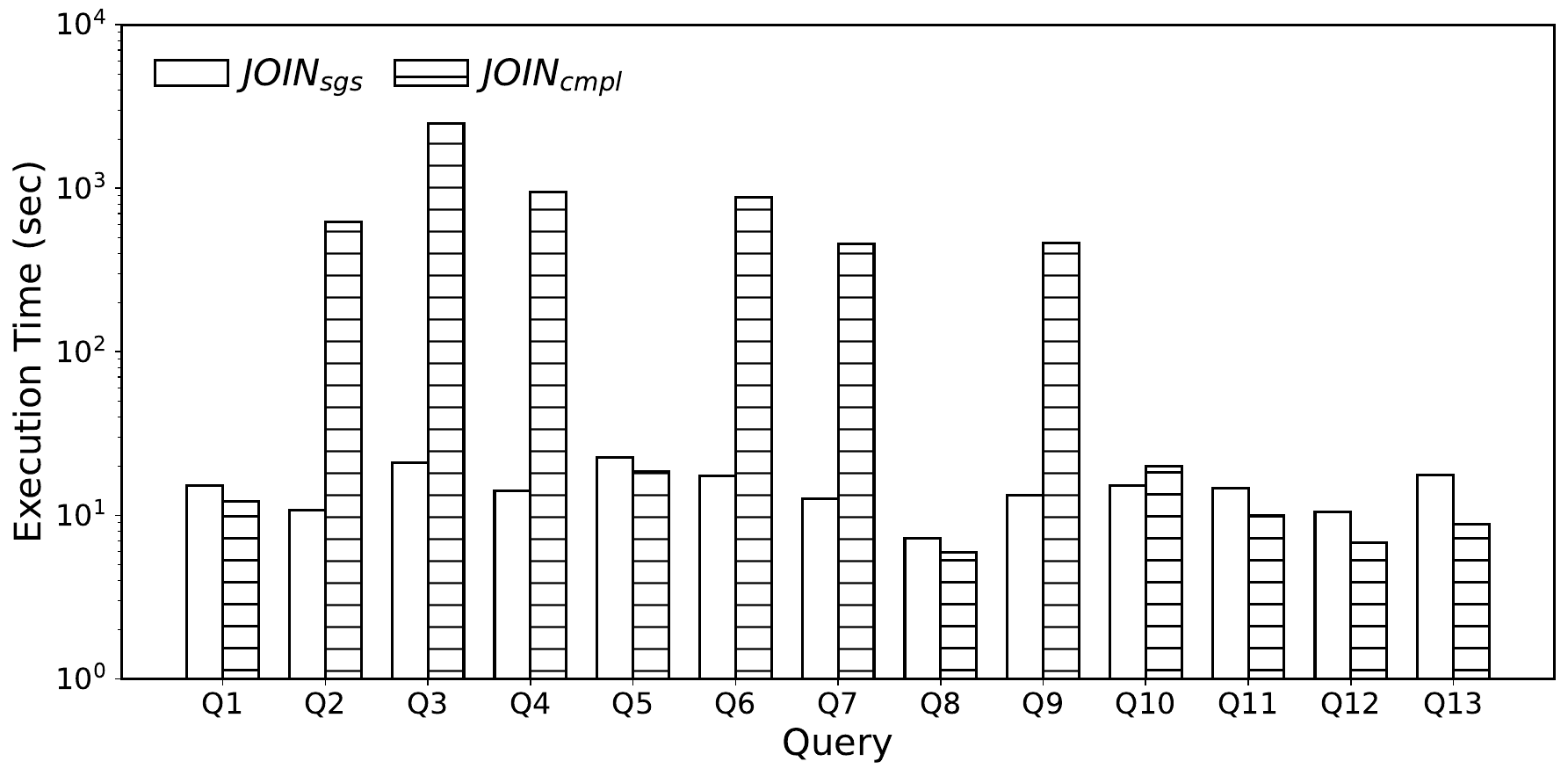}
        \caption{\joinsample vs. \joinnosample}
        \label{fig:exp:joinsample_sf3}
    \end{subfigure}
    \caption{\centering  Evaluate \sourcesample and \joinsample on SF3}
    \label{fig:exp:operators_sf3}
\end{figure}


\begin{figure}[t]
    \centering
    \begin{subfigure}[b]{0.45\textwidth}
        \centering
        \includegraphics[width=\textwidth]{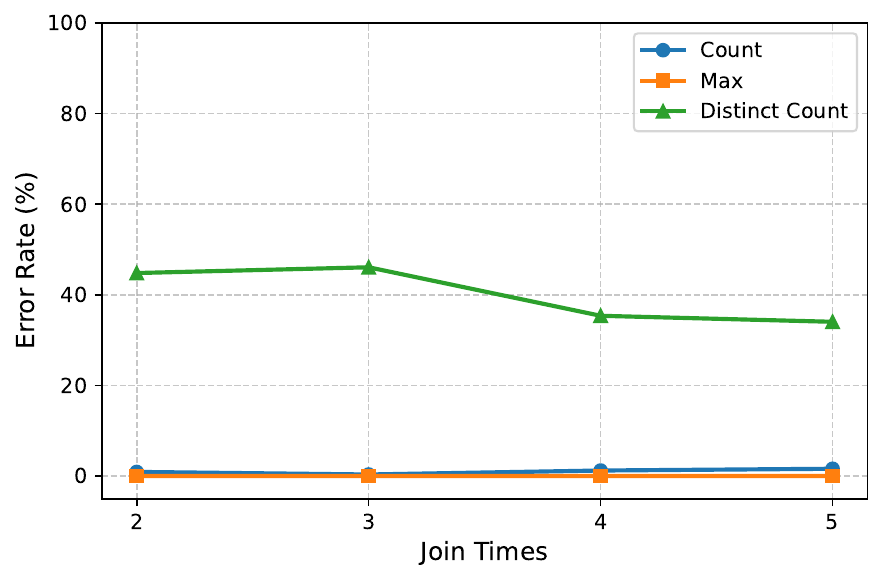}
        \caption{Cycle pattern}
        \label{fig:cycle}
    \end{subfigure}
    \hfill
    \begin{subfigure}[b]{0.45\textwidth}
        \centering
        \includegraphics[width=\textwidth]{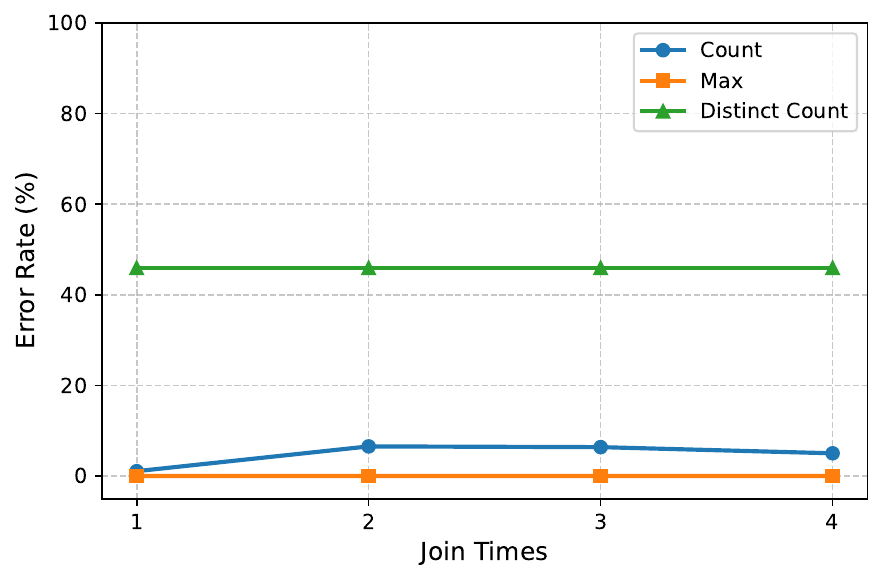}
        \caption{Path pattern}
        \label{fig:path}
    \end{subfigure}
    \caption{\centering Effect of multiple join iterations on accuracy}
    \label{fig:joinhop}
\end{figure}

\subsection{Ablation experiments on the \sourceop and \joinop}
\label{sec:exp:ablation}

To provide a more detailed performance analysis, we evaluate the \sourcesample~and \joinsample~operators individually on the SF3 dataset, comparing them with their non-sampling counterparts (\sourcenosample and \joinnosample). 
For fair comparison, \joinsample and \joinnosample use the same input hypergraphs obtained by \sourcesample with a sample size of 1M.
\reffig{exp:operators_sf3} shows the execution times of both sampling-based and non-sampling operators. 

\stitle{Source Operator.} As shown in \reffig{exp:sourcesample_sf3}, \sourcesample successfully completes all 13 queries while \sourcenosample fails on 7/13 queries with \texttt{OOM} errors (i.e., $Q_1$, $Q_2$, $Q_3$, $Q_5$, $Q_6$, $Q_9$, $Q_{10}$). For \revisesigmod{the queries where both} succeed, \sourcesample often achieves competitive or superior execution times. For example, $Q_7$ (31.5s vs. 42.8s), $Q_8$ (8.2s vs. 19.9s), and $Q_{12}$ (6.6s vs. 12.9s). However, for \revisesigmod{the queries with small result} sizes such as $Q_{11}$ and $Q_{13}$, \sourcenosample is faster (29.4s vs. 105.4s for $Q_{11}$; 58.6s vs. 215.3s for $Q_{13}$) because the sampling overhead exceeds the benefit when results are small.

\stitle{Join Operator.} As shown in \reffig{exp:joinsample_sf3}, \joinsample~significantly outperforms \joinnosample~for queries with large join results, achieving 35--119$\times$ speedups. For example, $Q_3$ achieves 119$\times$ speedup (20.9s vs. 2492.5s), $Q_4$ achieves 67$\times$ speedup (14.1s vs. 949.5s), and $Q_7$ achieves 36$\times$ speedup (12.6s vs. 456.7s). However, for queries with smaller join results such as $Q_1$, $Q_5$, $Q_8$, $Q_{11}$, $Q_{12}$, and $Q_{13}$, both approaches show comparable performance because the join computation itself is already efficient when the result size is small.

Overall, the results suggest that both \sourcesample~and \joinsample are essential for executing complex queries on large-scale datasets, with \sourcesample enabling completion of queries that would otherwise fail, and \joinsample providing significant speedups for queries with large intermediate results.

\begin{figure*}[tbp]
    \centering
    \includegraphics[width=\textwidth]{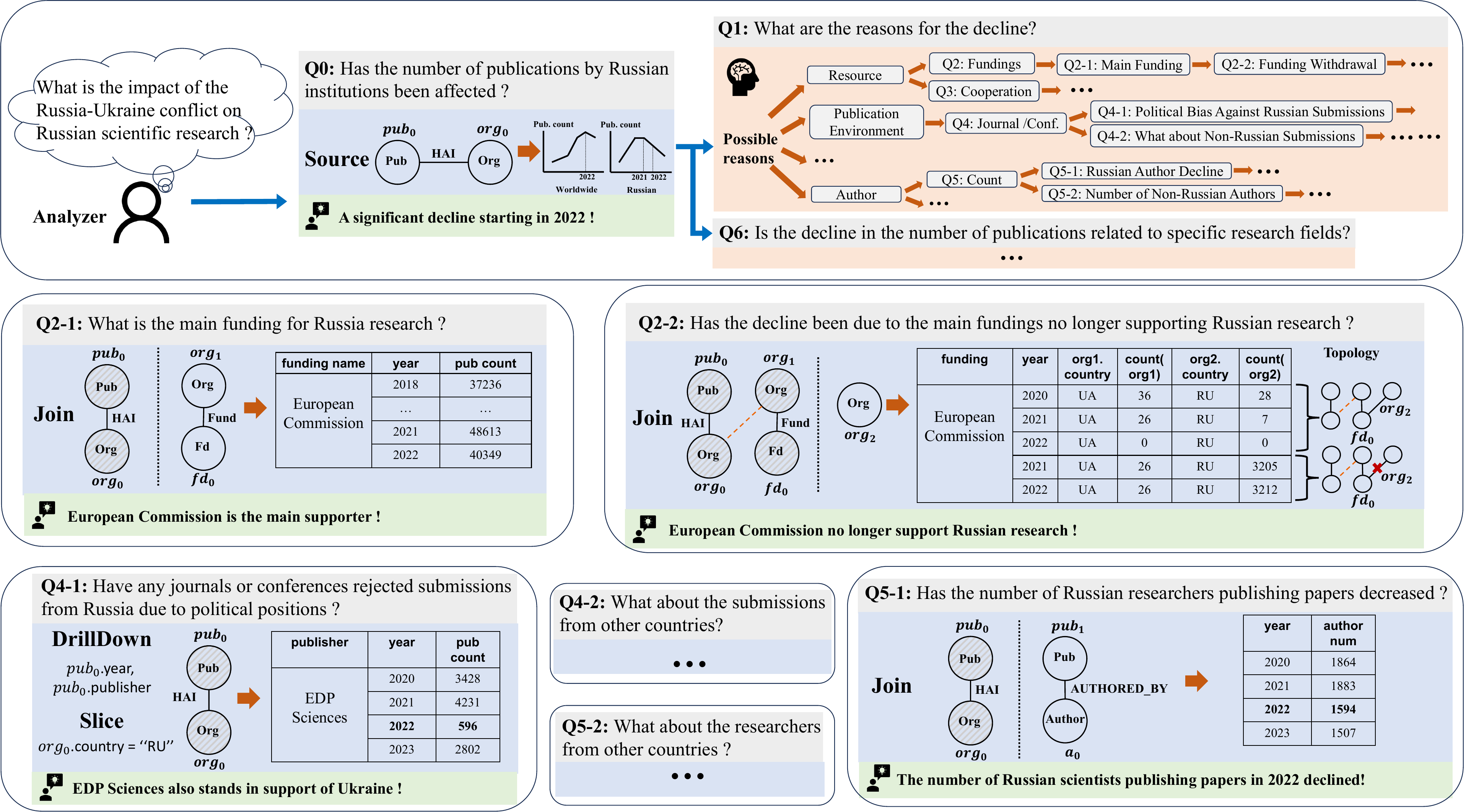}
    \caption{Case Study: Exploring the impact of the Russia-Ukraine conflict on Russian academia. We abbreviate Publication, Organization, and Funding as Pub, Org, and Fd, respectively. Shaded vertices indicate hypergraphs that have already been computed and can be reused in subsequent queries in this exploration.}
    \label{fig:case_study}
\end{figure*}


\subsection{Effect of multiple \joinop operations on accuracy}
\label{sec:exp:multi_join}

An important concern for sampling-based approaches is whether errors accumulate across multiple \joinop operations. We investigate this by conducting experiments on SF0.1 using cycle patterns (up to 6 edges, with 2--5 join iterations) and path patterns (up to 5 edges, with 1--4 join iterations). We use SF0.1 because Neo4j encounters \texttt{OOT} errors when finding 6-edge cycles on SF1.

To construct these patterns, we adopt an incremental construction approach with \sourceop and \joinop operators. For the 5-path, we start from a specific source vertex (e.g., user with ID 0) and retain only paths originating from that vertex, because Neo4j encounters \texttt{OOT} errors when enumerating all 5-paths without such constraints even on SF0.1. For the cycle pattern, we use a similar incremental approach without the starting vertex constraint. We measure the error rates of \distinctcountfunc, \maxfunc, and \countfunc aggregate functions, with results shown in \reffig{joinhop}.

Results show that error rates remain stable across join iterations, indicating no systematic error accumulation. 
For cycle patterns (2--5 joins), \countfunc error rates range from 0.35\% to 1.62\%, \maxfunc maintains near-zero error rates (~0.01\%), and \distinctcountfunc stays within 34--46\%. 
For path patterns (1--4 joins), \countfunc error rates range from 1.1\% to 6.6\%, \maxfunc achieves near-zero error rates (~0.01\%), and \distinctcountfunc shows stable error rates (~46\%). 
The relatively high error rates for \distinctcountfunc are expected due to the use of a biased estimator, as discussed in \refsec{exp:sample_size}.
Importantly, no aggregate function shows systematic error increase as the number of joins grows, demonstrating effective error control across multiple join operations---critical for complex multi-hop queries in \explorativebi scenarios.

\section{Case Study}
\label{sec:exp:case-study}

We demonstrate \sys through a comprehensive analysis of the question 
``What is the impact of the Russia-Ukraine conflict on Russian scientific research?'' on the OpenAIRE dataset. The details are shown in \reffig{case_study}. 

\stitle{Phase 1: Case Establishment.}
To investigate whether the Russia-Ukraine conflict affected Russian scientific research, we first need to establish the case. We use \sourceop to construct a Publication-Organization hypergraph, then apply \drilldownop on both \texttt{org.country} and \texttt{pub.year} dimensions. This yields the number of publications per country per year, revealing a notable decline in Russian institutions' output since 2022. However, this alone does not confirm causation. To verify this is Russia-specific rather than a global trend, we perform \drilldownop only on \texttt{pub.year} (without the country dimension), obtaining global publication counts. 
The comparison shows that worldwide publications increased in 2022, while Russian output dropped significantly---strongly suggesting the decline is linked to the conflict.

\stitle{Phase 2: Divergent and In-depth Exploration.}
Having established the case, follow-up questions naturally arise, such as \textit{Q1: What are the reasons for the decline?} and \textit{Q6: Is the decline related to specific research fields?} Taking Q1 as an example, we explore possible causes including \textit{resource factors}, \textit{publishing environment}, and \textit{author-related factors}, and conduct in-depth exploration for each.

\textit{Resource Factors.} Relevant resources may include funding, international cooperation, etc. Taking funding as an example, we first explore what are the main funding sources for Russian research (Q2-1). We extend the hypergraph via \joinop with Funding data---note that this reuses the Publication-Organization hypergraph obtained in Q0. Query Q2-1 reveals that the European Commission (EC) was a major funder of Russian institutions. 

Next, we investigate whether EC's funding for Russian research changed around 2022 (Q2-2). We extend Q2-1's result with another Organization entity to examine whether EC ceased funding Russia while continuing to support Ukraine. The results show that in 2022, EC ceased funding projects involving both Russian and Ukrainian institutions simultaneously (the corresponding row shows zero), while continuing to fund Ukraine-only projects (26 organizations). This suggests EC's policy shift, which is confirmed by official EC documentation~\cite{european_commission_website}.

\textit{Publishing Environment.} \revisesigmod{Here, publishing environment primarily refers to the attitudes of journals and conferences towards Russian research}. 
To explore this factor, we propose Q4-1, reusing the hypergraph obtained in Q0 and applying \drilldownop on the \texttt{pub.year} and \texttt{pub.publisher} dimensions. 
\revisesigmod{Then, a \sliceop is applied with the condition \texttt{org.country = "RU"} to examine how many papers each publisher published from Russian institutions per year.}
The results reveal that EDP Sciences dramatically reduced publications from Russian institutions after 2022. Investigating further, we found that EDP Sciences publicly announced support for Ukraine\footnote{\url{https://www.edpsciences.org/en/news-highlights/2524-edp-sciences-statement-on-support-for-ukraine}}, reflecting an increasingly challenging publishing environment for Russian researchers. Furthermore, to verify whether EDP Sciences specifically reduced only Russian publications in 2022, we conduct Q4-2 for additional validation.

\textit{Author Factors.} Similarly, we consider author-related factors, which may include author count, productivity, etc. In this example, we focus on author count and compare changes in the number of Russian authors who published papers (Q5-1) versus authors from other countries (Q5-2) to determine whether author factors contribute significantly to the decline in Russian publications. 
Specifically, for Q5-1, we introduce the Author dimension into the hypergraph obtained from Q0. The query results reveal that the number of Russian researchers who published papers in 2022 decreased by 15\% compared to 2021, suggesting that author factors may be a contributing cause to the decline in Russian publications.

\stitle{Benefits of \sys.}
This case study demonstrates \sys's core advantages: 
(1) \textbf{Reduced expertise dependency}: analysts explored hypotheses iteratively rather than having to define all factors upfront. Furthermore, building agentic systems atop \sys to assist in query formulation is a key direction for future work; 
(2) \textbf{Dynamic schema evolution}: we progressively incorporated funding, publisher, and author dimensions without rebuilding the base model; (3) \textbf{Result reusability}: the Publication-Organization hypergraph was reused across all analyses.
These advantages effectively \revisesigmod{address three major limitations identified in \refsec{intro} for \kw{Exploratory} \kw{BI}.
In terms of \textbf{computational efficiency}}, we have demonstrated in \refsec{experiment} that \sys efficiently handles complex queries by leveraging sampling techniques.

\section{Conclusion}
\label{sec:conclusion}

In this paper, we address critical limitations of traditional BI systems—high dependency on analyst expertise, static schemas, computational costs, and lack of reusability—by introducing \explorativebi, a novel paradigm enabling iterative, incremental analysis. We present \sys, a system built on a Hypergraph Data Model with operators (\sourceop, \joinop, \viewop, and OLAP operators) that support dynamic schema evolution and intermediate result reuse. To achieve real-time responsiveness, we implement sampling-based optimizations with provable theoretical guarantees for the estimation. Experimental results on the LDBC datasets demonstrate that \sys achieves significant speedups: on average 16.21$\times$ compared to Neo4j and 46.67$\times$ compared to MySQL, while maintaining high accuracy with an average error rate of only 0.27\% for \countfunc. 

\bibliographystyle{spmpsci}      
\bibliography{bi-sample}   


\end{document}